\documentclass[final]{siamltex}

\usepackage{amsfonts}
\usepackage{graphicx}
\usepackage{subfig}
\usepackage{booktabs}

\usepackage{amsmath}
\usepackage{amssymb}
\usepackage{epsf}

\usepackage{multirow}

\usepackage{textcomp}

\newtheorem{prop}{Proposition}

\DeclareFontFamily{OT1}{pzc}{}
\DeclareFontShape{OT1}{pzc}{m}{it}{<-> s * [1.10] pzcmi7t}{}
\DeclareMathAlphabet{\mathpzc}{OT1}{pzc}{m}{it}

\usepackage{xcolor}

\begin{document}

\thispagestyle{empty}
\bibliographystyle{siam}

\title{Total communicability\\ as a centrality measure}

\author {
Michele Benzi\thanks{Department of Mathematics and Computer
Science, Emory University, Atlanta, Georgia 30322, USA
(benzi@mathcs.emory.edu). The work of this author was supported
by National Science Foundation grant
DMS1115692.} \and
Christine Klymko\thanks{Department of Mathematics and Computer Science,
Emory University, Atlanta, Georgia 30322, USA
(cklymko@emory.edu).}
}

\maketitle

\markboth{{\sc M.~Benzi and C.~Klymko}}{{Total communicability}}

\begin{abstract}
We examine node centrality measures based on the notion
of {\em total communicability}, defined in terms of the 
row sums of matrix functions of the adjacency matrix of the network.
Our main focus is on the matrix exponential and the resolvent, which
have natural interpretations in terms of walks on the underlying
graph.
While such measures have been used before for ranking nodes in a network,
we show that they can be computed very rapidly even in the 
case of large networks. Furthermore, we propose the (normalized) total sum
of node communicabilities as a useful measure of network 
connectivity. Extensive numerical studies are conducted
in order to compare this centrality measure with the closely related
ones of subgraph centrality [E.~Estrada and J.~A.~Rodr\'iguez-Vel\'azquez,
Phys.~Rev.~E, 71 (2005), 056103] and Katz centrality 
[L.~Katz, {\em Psychometrica}, 18 (1953), pp.~39--43]. Both synthetic
and real-world networks are used in the computations.
\end{abstract}



\section{Introduction}
Over the past several years, the analysis of networks has become increasingly 
important in a number of disciplines 
\cite{Caldarelli,crofootetal,Ebook11,Ebook10,EHB11,Pagerank,New03,dhillon}.  
Network analysis 
is used in many situations: from determining network structure and 
communities, to describing the interactions between various elements 
of the network, to investigating the dynamics of phenomena taking place
on the network (e.g., information flow).

One of the fundamental questions in network analysis is to determine 
the ``most important" elements in a given network. Measures of node importance 
are usually referred to as {\em node centrality}, and many centrality measures 
have been proposed, starting with the simplest of all, the node degree. This
crude metric has the drawback of being too ``local", as it does not take
into effect the connectivity of the immediate neighbors of the node under
consideration. A number of more sophisticated centrality measures have
been introduced that take into account the global connectivity properties
of the network. These include various types of eigenvector centrality
for both directed and undirected networks, 
betweenness centrality, and others which are discussed below. 
 Overviews of various centrality measures can be found in 
\cite{Bocca06,Bonacich,Brandes,Ebook11,IRSurvey,New10,NewBarWat03}.  The 
centrality scores can be used to provide {\em rankings} of the 
nodes in the network. There are many different ranking methods 
in use (most of which depend on centrality measures), and many 
algorithms have been developed to compute these rankings.  
Information about the many different ranking schemes can be found, 
e.g., in \cite{BEK,Ebook11,Katz,HITS,IRSurvey,Pagerank,Whos,SALSA}.

One now standard method of measuring node importance is {\em subgraph 
centrality} \cite{estradarodriguez05}, which is based on
the diagonal entries of a matrix function applied to the adjacency 
matrix $A$ of the network in question.  Here, the matrix exponential
$e^A$ is frequently used.  While this approach has been successfully
used in a number of problems \cite{Ebook11,EHB11,NetworkProp},
 obtaining estimates of the 
diagonals of $e^A$ for a large network with adjacency matrix $A$ can be 
quite expensive. Indeed, computing individual entries of matrix 
functions $f(A)$ is generally costly for large $A$ even with the 
best available algorithms \cite{BenziBoito,EHB11}. 

In recent years, efficient algorithms have been developed for
computing the {\em action} of a matrix function on a vector, that is,
for computing the vector $f(A){\bf v}$ for a given matrix $A$ (usually large
and sparse), vector $\bf v$, and function $f$. A particularly important
case is that of the matrix exponential, since this provides a
solution method for initial value problems for first-order systems 
of linear ordinary differential equations. These algorithms, based
on variants of the Lanczos, Arnoldi or other Krylov subspace method,
access the matrix $A$ only in the form of (sparse) matrix-vector
products and have $O(n)$ storage cost for a sparse $n\times n$ matrix $A$
\cite[Chapter 13]{highambook}. When ${\bf v} = {\bf 1}$, the
vector with all its entries equal to 1, the $i$th entry
of the resulting
vector $f(A){\bf 1}$ contains the $i$th row sum of $f(A)$:
$$\left [ f(A){\bf 1} \right ]_i = \sum_{j=1}^n \left [ f(A) \right ]_{ij},
\quad 1\le i \le n\,.$$
This quantity, which has a graph-theoretic interpretation 
in terms of subgraph centrality and communicability \cite{estradahatano2,EHB11}, 
can be computed much faster than subgraph centrality using current
computational techniques. Of course, the same is true if the vector
${\bf 1}$ is replaced by some other vector---typically, an \lq\lq external
importance vector\rq\rq\ which can be used to take into account intrinsic,
not network-related contributions to the centrality of each node 
\cite[pp.~174--175]{New10}.

Such centrality measures have long been in use in network analysis. 
Note that for the case of the \lq\lq identity function\rq\rq\
$f(A) = A$, and symmetric $A$ (undirected networks),
we recover degree centrality. The off-diagonal row sums of $e^A$ 
have been used in social network analysis to measure the resilience 
of an individual in the face of hostile attacks from within the network 
\cite[Chapter 6]{Ebook11}. More recently, row and colums sums of $e^A$
have been applied to the
identification of hubs and authorities in directed networks \cite{BEK}. 
For resolvent-type functions,
such as $f(A) = (I-\alpha A)^{-1}$ (with $I$ the $n\times n$
identity matrix), for suitable values of
$\alpha >0$, we recover the well-known {\em Katz centrality}
and its variants, also known as $\alpha$-centrality; see, e.g.,
\cite{Katz} and \cite{Bonacich,BL01,BE06,GH13}. None of these
previous studies, however, considered algorithmic aspects such
as computational cost, storage, and so forth.

This paper considers the implications of using the row sums of $e^A$ or
similar matrix functions as a 
measure of node centrality, focusing for the sake of 
brevity on undirected networks.  The 
interpretation of this measure in terms of total
communicability of a node is given, and compared to the 
one for subgraph centrality in section \ref{sec:diagonal_vs_sum}.  
In section \ref{sec:total_network}, the concept of 
{\em total network communicability}
is introduced and discussed.
Section \ref{sec:experiments} contains experimental comparisons 
of subgraph centrality and total communicability using
various synthetic and real-world networks.  Sections \ref{sec:approx} 
and \ref{sec:resolvent} discuss computational aspects 
and the use of row sum centrality with other standard matrix functions, 
respectively.
We offer some conclusive remarks in section \ref{sec:concl}.

\section{Background and definitions}
The analysis of networks requires the use of notions
from graph theory, linear algebra, numerical analysis, 
and computer science.  Here we list some basic  
definitions and ideas from graph theory.  A more complete overview 
can be found in \cite{Die00}.  

A {\em graph} $G=(V,E)$ is a set of nodes (vertices) $V$ with 
$|V|=n$ and edges $E = \{(i,j)| i,j \in V\}$.  A graph is {\em undirected} 
if the edges are unordered pairs of vertices and {\em directed} 
if the pairs are ordered (edges have a direction).  The {\em degree} 
of a vertex in an undirected graph is the number of edges which 
are adjacent to the node.  In a directed graph, nodes have both 
an {\em in-degree}, the number of edges pointing into the node, 
and an {\em out-degree}, the number of edges starting at the node 
and pointing away.  A {\em simple} graph is a graph with no 
{\em loops} (edges from node $i$ to itself), no multiple edges, 
and unweighted edges.  In this paper, all networks correspond 
to simple, undirected graphs unless otherwise specified.

A {\em walk} of length $k$ on a graph $G$ is a sequence of 
vertices $v_1, v_2, \ldots, v_{k+1}$ such that $(v_i,v_{i+1}) 
\in E$ for all $1 \leq i \leq k$.  A {\em path} is a walk 
with no repeated vertices.  A {\em closed walk} is a 
walk that starts and ends at the same vertex.  A {\em cycle} 
is a closed walk with no repeated vertices. If any vertex 
in the graph is reachable from any other vertex,
the graph is said to be {\em connected}.

Every graph can be viewed as a matrix through the use of its 
{\em adjacency matrix}.  The adjacency matrix of a network 
with graph $G$ is given by 
$$A = (a_{ij}); \quad a_{ij}=\left\{\begin{array}{ll}
1,& \textnormal{ if } (i,j) \textnormal{ is an edge in } G,\\
0, & \textnormal{ else. }
\end{array}\right .
$$
The requirement of unweighted edges causes $A$ to be binary 
and that of no loops in the graph forces $A$ to have zeros 
along its diagonal.  If the network is undirected, $A$ will 
be symmetric but if the network is directed, $A$ will generally 
be unsymmetric.  In the case of an undirected network, the 
eigenvalues of $A$ will be real. We label the eigenvalues of $A$
in non-increasing order: 
$\lambda_1 \geq \lambda_2 \geq \ldots \geq \lambda_n$. Note that
the Perron--Frobenius theorem implies that $\lambda_1 > \lambda_2$ if the graph is
connected (equivalently, if $A$ is irreducible).

\section{Diagonal entries vs.~row sums}
\label{sec:diagonal_vs_sum}
In \cite{estradarodriguez05}, the authors introduce the concept of subgraph 
centrality as a centrality measurement for nodes in a network.  
This provides a ranking based on the diagonal entries of a matrix 
function applied to the adjacency matrix.  Although there are various 
choices of function to use, the most common is the matrix exponential.  
The {\em subgraph centrality} of node $i$ is given by $[e^A]_{ii}$ 
where $A$ is the adjacency matrix of the network.  
The {\em subgraph communicability} between nodes $i$ and $j$ 
is given by $[e^A]_{ij}$ (note that in the case of an undirected 
network, $A$ is symmetric and $[e^A]_{ij}=[e^A]_{ji}$).  A node 
with a (relatively) large subgraph centrality is considered to be more important 
in the network and is given a higher ranking than nodes with lower 
subgraph centrality.  A (relatively) large subgraph communicability between a pair 
of nodes $i$ and $j$ indicates that information flows more easily between 
those two nodes than between pairs of nodes with lower communicability.  
In other words, a low subgraph communicability indicates that the two 
nodes cannot easily exchange information.  
Network communicability can also be interpreted in terms of the correlations between
different components of physical systems; see, e.g., \cite{EHB11}. 

The reasoning behind using the diagonal entries of $e^A$ as a measure 
of the centrality of a node in the network can be seen by considering 
the power series expansion of $e^A$ \cite{highambook}:

\begin{equation}
e^A = I + A + \frac{A^2}{2!} + \frac{A^3}{3!} + \cdots + 
\frac{A^k}{k!} + \cdots = \sum_{k=0}^{\infty} \frac{A^k}{k!}\,.
\end{equation}

It is well known in graph theory (and fairly easy to prove) that 
if $A$ is the adjacency matrix of a network with unweighted 
edges, then $[A^k]_{ij}$ counts the number of walks of 
length $k$ between nodes $i$ and $j$.  Thus, the subgraph 
centrality of node $i$, which is equal to $[e^A]_{ii}$, counts 
the number of closed walks centered at node $i$ weighting 
a walk of length $k$ by a penalty factor of $\frac{1}{k!}$.  
In this way, shorter walks are deemed more important than 
longer walks.  Although some of these walks can be described 
as ``illogical'' (for example, the walk $v_i \rightarrow v_j 
\rightarrow v_i \rightarrow v_j \rightarrow v_i$ is a closed walk 
of length 4 centered at node $i$), the subgraph centrality of 
node $i$ still gives us a measure of how close node $i$ is to 
everything else in the network.

By contrast, the row sum of $e^A$ for node $i$ is given by 
$\sum_{j=1}^n [e^A]_{ij}$, which counts all walks between 
node $i$ and all the nodes in the network (node $i$ included), 
weighting walks 
of length $k$ by a penalty factor of $\frac{1}{k!}$.  
Thus, the $i$th row sum of $e^A$ can be interpreted as the {\em total
subgraph communicability} of node $i$, and can be interpreted
as a measure of the importance of the $i$th node in the network, 
since a node with high communicability with a large number of
other nodes in the network is likely to be an important node,
and certainly a more important node than one characterized by
low total communicability. 

An immediate question is how this centrality measure compares with
the subgraph centrality 
of node $i$ in the network. In general, the rankings produced by 
the total communicability measure will not 
be the same as those produced by the subgraph centrality measure.  
The difference between the two rankings is
\begin{equation}
\label{eq:centrality_diff}
\sum_{j=1}^n [e^A]_{ij} - [e^A]_{ii} = 
\sum_{j\neq i} [e^A]_{ij} = \sum_{j \neq i} 
\sum_{k=1}^n e^{\lambda_k} v_{ki}v_{kj}\,,
\end{equation}
where $v_{ik}$ is the $i$th element of the normalized eigenvector 
${\bf v}_k$ of $A$ associated with the eigenvalue $\lambda_k$. 
Note that $e^A$ is always positive definite and that its diagonal
entries are often large compared to the off-diagonals. 
If the diagonal entries of $e^A$ vary over a wide range while its
off-diagonal sums remain confined within a more narrow range,
the rankings produced by the 
two methods will not differ by much.  
However, this depends both on the spectrum of $A$ and the entries 
of the eigenvectors.

While it appears to be difficult, in general, to establish 
a relation between the rankings 
produced by the subgraph centrality and total 
communicability, for certain types of simple graphs it is 
easy to show that the two methods will give identical rankings.  
These include {\em complete graphs} and 
{\em cycles} (where each node has the exact same ranking under 
both systems), {\em paths} and {\em star graphs}.  A star graph on $n$ nodes 
has one central node that is connected to each of the $n-1$ 
remaining nodes and no other edges.  Under both ranking 
systems, the central node is ranked highest and the remaining 
nodes all have the same scores.  This can be shown either 
using graph theory or by examining the eigenvalues and 
eigenvectors of the star graph (more information about 
the spectra of star graphs can be found in \cite{stareig}).

One case where the two measures could be expected to give similar
rankings is that of networks with a large {\em spectral gap}, which 
for the purposes of this paper is the difference $\lambda_1 - \lambda_2$
between the first (largest) and second eigenvalue. We have:
$$\left [e^A\right ]_{ii} = e^{\lambda_1} v_{1i}^2 + \sum_{k=2}^n  e^{\lambda_k}v_{ki}^2$$
and 
$$\left [e^A{\bf 1}\right ]_i = e^{\lambda_1} 
\left ({\bf v}_1^T{\bf 1}\right )v_{1i}
+ \sum_{k=2}^n  e^{\lambda_k}\left ({\bf v}_i^T{\bf 1}\right )v_{ki}.$$
Dividing both expressions by the constant $e^{\lambda_1}$ (which does
not affect the rankings) and observing\footnote{By the Perron--Frobenius 
Theorem, the dominant
eigenvector can be chosen to have nonnegative entries, and
positive entries when the
graph $G$ is connected.} 
that
${\bf v}_1^T{\bf 1} = \|{\bf v}_1\|_1$ shows that for $\lambda_1 \gg \lambda_2$ the
two rankings are largely determined by the quantities
$v_{1i}^2$ and $\|{\bf v}_1\|_1 v_{1i}$,
respectively, and therefore by the entries $v_{1i}$ of the dominant
eigenvector of $A$. 
Thus, 
if the difference $\lambda_1 - \lambda_2$ is sufficiently large,
the two centrality measures reduce to  
eigenvector centrality \cite{Bonacich} and therefore can be expected to result
in very similar rankings, especially for the top nodes.
Numerical experiments (not shown here) performed on Erd\"os--Renyi graphs 
with large spectral gaps have confirmed this fact.

However, it is difficult
to quantify {\em a priori} how large the spectral gap needs to be for all these
rankings to be identical (or even approximately the same).
In the section on computational experiments we will see that 
there can be significant differences between the rankings obtained
using subgraph centrality and those using
total communicability centrality, even
for networks with a relatively large spectral gap.

\section{Total network communicability}
\label{sec:total_network}
The total communicabilities of individual nodes give a measure 
of how well each node communicates with the other nodes of the network.  
In order to measure how effectively communication takes place 
across the network as a whole, we consider the sum of all 
the total communicabilities.  For a network with adjacency matrix 
$A$, this is given by 
\begin{equation}
C(A) = \sum_{i=1}^n \left [e^A{\bf 1}\right ]_i = \sum_{i=1}^n 
\sum_{k=1}^n  e^{\lambda_k}\left ({\bf v}_i^T{\bf 1}\right )v_{ki} = {\bf 1}^Te^A{\bf 1}\,,
\end{equation}
 where, as in section \ref{sec:diagonal_vs_sum}, $\lambda_k$ is the 
$k$th eigenvalue of $A$ and $v_{ik}$ is the $i$th element of the normalized 
eigenvector ${\bf v}_k$ associated with $\lambda_k$. 
Here we propose to use the {\em total 
network communicability}, $C(A)$, as a global measure of the ease of 
sending information across a network. We emphasize that while $C(A)$
is defined as the sum of all the entries of $e^A$, it is not necessary
to know any of the individual entries of $e^A$ to compute $C(A)$; indeed,
very efficient methods exist to compute quadratic forms of the type
${\bf v}^Tf(A){\bf v}$ for a given function $f(x)$, matrix $A$
and vector $\bf v$, see \cite{BenziBoito,BenziGolub,golubmeurantbook}.
 
It is instructive to compare the total communicability of a
network with the {\em Estrada index}, an important graph invariant
defined as the sum of all the subgraph centralities:
$$EE(A) = \sum_{i=1}^n \left 
[e^A\right]_{ii}= \sum_{i=1}^n e^{\lambda_i} = {\rm Tr}(e^A).$$  
The following proposition provides simple lower and upper 
bounds for $C(A)$ in terms of $EE(A)$ and other spectral
quantities associated with the underlying network.
 
\begin{prop}\label{prop1}
Let $A$ be the adjacency matrix of a simple
network on $n$ vertices.  Then, 
$$EE(A) \leq C(A) \leq n\,e^{\|A\|_2}\,,$$
\label{prop:tnc_exp}
where $\|A\|_2$ denotes the spectral norm of $A$.
In particular, for an undirected network we have
$$EE(A) \leq C(A) \leq n\, e^{\lambda_1}.$$
 \end{prop} 
 \begin{proof}
 The lower bound is trivial, as $$EE(A) = \sum_{i=1}^n \left 
[e^A\right]_{ii} \leq \sum_{i=1}^n \sum_{j=1}^n \left [e^A\right]_{ij} 
= \sum_{i=1}^n \left [e^A{\bf 1}\right ]_i = C(A).$$ The upper 
bound follows from noticing that $C(A) = {\bf 1}^Te^A{\bf 1} = 
\left(e^A{\bf 1}\right)^T{\bf 1} = \langle e^A{\bf 1}, {\bf 1} 
\rangle$ and applying the Cauchy--Schwarz inequality to the quadratic 
form $\langle e^A{\bf 1}, {\bf 1} \rangle$:  $$|\langle e^A{\bf 1}, {\bf 1} \rangle | 
\leq \|e^A{\bf 1}\|_2 \|{\bf 1}\|_2 \leq \|e^A\|_2 \|{\bf 1}\|_2 
\|{\bf 1}\|_2 \leq n\, e^{\|A\|_2}.$$
For an undirected network $A$ is symmetric and $\lambda_1 = \|A\|_2$.
\end{proof}

Note that
the lower bound is attained in the case of the \lq\lq empty\rq\rq\  graph with 
adjacency matrix $A =  0$, while the upper bound is attained on the 
complete graph, whose adjacency matrix is $A = {\bf 1}{\bf 1}^T - I$.  

The bounds from Proposition \ref{prop:tnc_exp} also hold 
for $e^{\beta A}$, $\beta > 0$.  For any connected graph with 
adjacency matrix $A$, the bounds get tighter as $\beta \rightarrow 0+$, 
since both the lower and upper bound tend to 1.
The parameter $\beta$ can be interpreted 
as an {\em inverse temperature} and is a reflection of external 
disturbances on the network (see, e.g., \cite{EHB11} for details);
taking $\beta \rightarrow 0+$ is equivalent to \lq\lq raising
the temperature\rq\rq\ of the environment surrounding the network.

When appropriately normalized,  
$C(A)$ can be used to compare 
the ease of information exchange on different networks. This
could be useful, for instance, in the design of communication networks.
In the following sections we compute the total communicability
for various types of networks. The question arises of what would constitute a 
reasonable normalization factor. There are several possibilities. Normalizing
$C(A)$ by the number $n$ of nodes corresponds to the average 
total communicability of the network per node. Similarly,
normalizing $C(A)$ by the number $m$ of edges would correspond to the
average total communicability of the network per edge. We note also
that the minimum value of $C(A)$ is $n$, corresponding to the empty
graph on $n$ nodes ($V=\emptyset$), while the maximum value is
$n^2e^{n-1}-n$, corresponding to the complete graph on $n$ nodes.
The expression 
$$\hat C(A) := \frac{C(A)-n}{n^2e^{n-1}-2n}$$
takes its values in the interval $[0,1]$, with $\hat C(A) =0$ for \lq\lq empty\rq\rq\ 
graphs (no communication can take place on such graphs) and $\hat C(A) =1$ on
complete graphs (for which the ease of communication between nodes is
clearly maximum). Unfortunately, the denominator in this expression grows so fast
that for most sparse graphs evaluating $\hat C(A)$ results in underflow. 

In the experiments below we chose to normalize $C(A)$ by $n$, the number
of nodes, and by $m$, the number of edges; for the network used in our
tests we found that comparing networks based on $C(A)/n$ or on $C(A)/m$ yields exactly
the same rankings, therefore we only include results for the former measure.

\section{Computational studies}
\label{sec:experiments}

In this section we carry out extensive centrality computations for a variety
of networks, with the aim of comparing subgraph centrality with total 
communicability centrality. In particular, we are interested in determining
if, or for what type of networks, the two centrality measures provide similar 
rankings. Moreover, for those networks where the two measures result in rankings
that differ significantly, we would like to obtain some insights on why this is
the case. Of course it would be desirable to know when one measure should be preferred
to the other, but this is a difficult problem since it is not easy to come up
with objective criteria for comparing ranking methods (see the discussion in
\cite[Chapter 16]{Whos}).  We will compare the two methods in terms
of computational cost in section \ref{sec:approx}.

To measure similarities between the rankings obtained with the two methods
we use (Pearson) correlation coefficients and the intersection distance method (see 
\cite{Faginetal} as well as
\cite{Boldi,Randomalpha}) on both the full set $V$ of nodes and on partial lists of nodes.  
The correlation coefficients are computed using lists of nodes in rank order.  The intersection 
distances are computed using the lists of subgraph centrality and total communicability values.  Given two ranked lists $x$ and $y$, the intersection distance between the two lists is computed in the following way: let $x_k$ and $y_k$ be the top $k$ ranked items in $x$ and $y$ respectively.  Then the
top $k$ intersection distance (or {\em intersection similarity}) is given by 
$${\rm isim}_k(x,y) := \frac{1}{k}\sum_{i=1}^k \frac{|x_i\Delta y_i|}{2i}$$
where $\Delta$ is the symmetric difference operator between the two sets.  If the lists are identical, then ${\rm isim}_k(x,y) = 0$ for all $k$.  If the two sequences are disjoint, then ${\rm isim}_k = 1$.
We denote by $\mathpzc{cc}$ the correlation coefficient between the two vector 
rankings, and by $\mathpzc{cc}_p$ the correlation coefficient between 
the top $p$\% of nodes under the two ranking systems. We denote by 
${\rm isim}_{p\%}$ the intersection distance between the top $p\%$ of nodes. 

Unless otherwise specified, all experiments were performed using Matlab version
7.9.0 (R2009b) on a MacBook Pro running OS X Version 10.6.8, a 2.4 GHZ Intel Core 
i5 processor and 4 GB of RAM. In this section, we use the Matlab built-in
function $\tt expm$ for computing the matrix exponential.

\subsection{Test matrices}
\label{sec:tests}
The synthetic examples used in the tests were produced using the CONTEST toolbox 
in Matlab \cite{contest,contest09}.  The graphs tested
were of two types: preferential attachment (Barab\'asi--Albert) model 
and small world (Watts--Strogatz) model. In CONTEST, these graphs 
and the corresponding adjacency matrices can be
built using the functions {\tt pref} and {\tt smallw}, respectively.

The preferential attachment model was designed to produce 
networks with scale-free degree distributions as well as the small 
world property \cite{prefattach}.  In CONTEST, preferential 
attachment networks are constructed using the command {\tt pref(n,d)} 
where $n$ is the number of nodes in the network and $d \geq 1$ is the 
number of edges each new node is given when it is first introduced to 
the network.  The network is created by adding nodes one by one (each 
new node with $d$ edges).  The edges of the new node connect to nodes 
already in the network with a probability proportional to the degree 
of the already existing nodes.  This results in a scale-free degree 
distribution.  Note that with this construction, the minimum degree 
of the network is $d$.
When $d>1$ this means that the network has 
no dangling nodes (nodes of degree 1), whereas in many real-life networks
one often observes a high number of dangling nodes.
In the CONTEST toolbox, the default value is $d=2$.

\begin{table}[t!]
\centering
\caption{Comparison, using the correlation coefficient, 
of rankings based on the diagonal entries 
and row sums of $e^A$ for 1000-node scale-free  
networks of various parameters built using the {\tt pref} function 
in the CONTEST Matlab toolbox.  The values reported are the 
averages over 20 matrices with the same parameters.  The 
parameter $d$ is the initial degree of nodes in the network 
(and consequently the minimum degree of the network).}
\begin{tabular}{|c|c|}
\hline
$d$  & $\mathpzc{cc}$ \\
 \hline
 \hline
 1 & 0.224 \\
 \hline
 2 & 0.343 \\
\hline
 3 & 0.517 \\
\hline
4 & 0.905 \\
\hline
5 & 0.993 \\
\hline
6 & 0.999 \\
\hline
7 & 0.999 \\
\hline
$\geq 8$ & 1 \\
\hline
\end{tabular}
\label{tbl:pref_exp_compare}
\end{table}

\begin{table}[t!]
\centering
\caption{Intersection distance comparisons of rankings based on the diagonal entries 
and row sums of $e^A$ for 1000-node scale-free  
networks of various parameters built using the {\tt pref} function 
in the CONTEST Matlab toolbox.  The values reported are the 
averages over 20 matrices with the same parameters.  The 
parameter $d$ is the initial degree of nodes in the network 
(and consequently the minimum degree of the network).}
\begin{tabular}{|c|c|c|}
\hline
$d$  & isim & ${\rm isim}_{10\%}$  \\
 \hline
 \hline
 1 & 0.174 & 0.199 \\
 \hline
 2 & 0.036 & 0.031 \\
\hline
 3 & 0.003 & 0.005 \\
\hline
4 & 2.04e-4 & 2.79e-4 \\
\hline
5 & 1.30e-5 & 1.71e-5 \\
\hline
6 & 9.83e-7 & 0 \\
\hline
7 & 4.93e-7 & 0 \\
\hline
$\geq 8$ & 0 & 0 \\
\hline
\end{tabular}
\label{tbl:pref_exp_compare_isim}
\end{table}

In our experiments, we tested various values of $d$ on a network of 
size $n=1000$: twenty networks were tested for all values 
$1 \leq d \leq 10$, as well as all a few larger values.  
In Table \ref{tbl:pref_exp_compare}, the averages of the correlation 
coefficients between the subgraph centrality rankings and the total 
subgraph communicability rankings can be found for various values of $d$.  
The intersection distance values can be found in Table \ref{tbl:pref_exp_compare_isim}.  The intersection distance values were calculated both for the full set of rankings and for the top 10\% of ranked nodes. 

The results show that correlation between the two metrics increases and the intersection distance value decreases quickly with the value of the parameter $d$. 
The intersection distance values for the top 10\% of nodes are very close to those for the complete set of nodes. For sufficiently dense networks,
the two measures provide essentially identical rankings, producing correlation coefficients close to 1 and intersection distances close to 0.

A second class of synthetic test matrices used in our experiments
corresponds to small-world networks (Watts--Strogatz model). 
The small world model was developed as a way to impose a high clustering coefficient 
onto classical random graphs \cite{smallw}.  The name 
comes from the fact that, like classical random graphs, the Watts--Strogatz 
model produces networks with the small world (that is, small graph diameter) property.  
To build these matrices, the input is {\tt smallw(n,d,p)} 
where $n$ is the number of nodes in the network, which are arranged in 
a ring and connected to their $d$ nearest neighbors on the ring.  
Then each node is considered independently and, with probability $p$, 
a link is added between the node and one of the other nodes in the network, 
chosen uniformly at random.  At the end of this process, all loops and 
repeated edges are removed.  For this set of experiments, the size of 
the network was fixed at $n=1000$ and the probability of an extra 
link was left at the default value of $p=0.1$ while $d$ was varied.  

The values of $d$ tested were: all values $1 \leq d \leq 10$, along with 
all multiples of 10 up to 200.  In each case, twenty networks were created 
with each value of $d$.  The average correlation coefficients between 
the subgraph centrality rankings and the total communicability 
rankings are given in Table \ref{tbl:smallw_exp_compare}. As before, the correlation 
coefficients were computed between the complete sets of rankings.  
The intersection distances, reported in Table \ref{tbl:smallw_exp_compare_isim},
 were computed on both the complete sets of 
rankings and the top 10\% of ranked nodes. 

\begin{table}[t!]
\centering
\caption{Comparison, using the correlation coefficient, of 
rankings based on the diagonal entries and row sums 
of $e^A$ for 1000-node small world networks of various parameters made 
using the {\tt smallw} function in the CONTEST Matlab toolbox.  
The values reported are the average over 20 matrices with the same parameters.}
  \subfloat[][]{\begin{tabular}{|c|c|}
\hline
$d$  & $\mathpzc{cc}$ \\
 \hline
 \hline
1 & 0.177 \\
\hline
2  & 0.089 \\
\hline
 3 & 0.037 \\
\hline
4 & 0.033 \\
\hline
5 & 0.031 \\
\hline
6 & 0.048 \\
\hline
7 & 0.039 \\
\hline
8 & 0.046 \\
\hline
9 & 0.031 \\
\hline
10 & 0.054 \\
\hline\end{tabular}}
\qquad
  \subfloat[][]{\begin{tabular}{|c|c|}
\hline
$d$  & $\mathpzc{cc}$ \\
 \hline
\hline
20  & 0.156 \\
\hline
 30 & 0.222 \\
\hline
40 & 0.240 \\
\hline
50 & 0.310 \\
\hline
60 & 0.426 \\
\hline
70 & 0.431 \\
\hline
80 & 0.747 \\
\hline
90 & 0.926 \\
\hline
100 & 0.997 \\
\hline
$\geq 110$ & 1 \\
\hline
\end{tabular}}
\label{tbl:smallw_exp_compare}
\end{table}

\begin{table}[t!]
\centering
\caption{Intersection distance comparison of rankings based on the diagonal entries and row sums 
of $e^A$ for 1000-node small world networks of various parameters made 
using the {\tt smallw} function in the CONTEST Matlab toolbox.  
The values reported are the averages over 20 matrices with the same parameters.}
  \subfloat[][]{\begin{tabular}{|c|c|c|c|c|}
\hline
$d$  & isim & ${\rm isim}_{10\%}$  \\
 \hline
 \hline
1 & 0.015 & 0.071 \\
\hline
2  & 0.056 & 0.160 \\
\hline
 3 & 0.089 & 0.252 \\
\hline
4 & 0.117 & 0.350 \\
\hline
5 & 0.151 & 0.479 \\
\hline
6 & 0.178 & 0.621 \\
\hline
7 & 0.218 & 0.709 \\
\hline
8 & 0.243 & 0.731 \\
\hline
9 & 0.262 & 0.705 \\
\hline
10 & 0.284 & 0.725 \\
\hline\end{tabular}}
\qquad
  \subfloat[][]{\begin{tabular}{|c|c|c|c|c|}
\hline
$d$  & isim & ${\rm isim}_{10\%}$  \\
 \hline
\hline
20  & 0.311 & 0.713  \\
\hline
 30 & 0.239 & 0.535  \\
\hline
40 & 0.133 & 0.351 \\
\hline
50 & 0.111 & 0.214 \\
\hline
60 & 0.039 & 0.120 \\
\hline
70 & 0.014 & 0.041 \\
\hline
80 & 0.002 & 0.007  \\
\hline
90 & 1.71e-4 & 4.05e-4  \\
\hline
100 & 5.88e-6 & 1.09e-5 \\
\hline
$\geq 110$ & 0 & 0 \\
\hline
\end{tabular}}
\label{tbl:smallw_exp_compare_isim}
\end{table}

It is evident from these results that for this class of small world
networks, the similarity between the two ranking measures is much
weaker than for the preferential attachment model, at least as long
as the networks remain fairly sparse. 
The intersection distances are also relativelt large, further indicating 
that the two measures are much more weakly related than in the case of 
the preferential attachment model. For some values of $d$, the intersection 
distance between the top 10\% of nodes is above 0.7, indicating that 
there is little consistency among the rankings of the top 10\% of nodes under 
the two measures.  As the networks become increasingly
dense, however, the correlation between the two measures becomes
stronger and the intersection distance eventually decreases. 

\subsection{Total communicability in small world networks}
For networks with low connectivity (or high locality), the total network communicability 
can be expected to be low compared with networks with higher connectivity.
For instance, on a 5000 node ring lattice, 
the total network 
communicability is $C(A) =$ 3.69e04 and the normalized $C(A)$ is 7.4.  
However, when even a few shortcuts are added across the lattice using the 
Watts--Strogatz small world model, this value jumps considerably.  
If the probability of a shortcut is $p=0.1$, the normalized total 
network communicability (averaged over 20 networks created using 
input {\tt smallw(5000,1, p)}) is 9.7.  If the probability of a 
shortcut is increased to $p=0.2$, the normalized total network 
communicability increases to 12.4.  These and additional
results can be found in Table \ref{tbl:avgtnc} and Fig.~\ref{fig:norm_C(A)}. 

\begin{table}
\centering
\caption{Comparison of the total network communicability $C(A)$ of a 
ring lattice and small world rings with increasing probability of a shortcut. 
The computed values 
were averaged over 20 instances.}
\begin{tabular}{|c|c|c|c|}
\hline
Graph & number of edges & $C(A)$ & normalized $C(A)$ \\
 \hline
 \hline
5000 node ring lattice & 5000 & 3.69e04 & 7.4 \\
 \hline
 {\tt smallw(5000,1,.1)} & 5492 & 4.83e04 & 9.7 \\
 \hline
 {\tt smallw(5000,1,.2)} & 6222 & 6.22e04 & 12.4 \\
 \hline
  {\tt smallw(5000,1,.3)} & 6495 & 7.92e04 & 15.8 \\
 \hline
  {\tt smallw(5000,1,.4)} & 6990 & 9.90e04 & 19.8 \\
 \hline
  {\tt smallw(5000,1,.5)} & 7496 & 1.24e05 & 24.8 \\
 \hline
  {\tt smallw(5000,1,.6)} & 7999 & 1.53e05 & 30.6 \\
 \hline
\end{tabular}
\label{tbl:avgtnc}
\end{table}

\begin{figure}[t!]
\centering
\includegraphics[width=0.6\textwidth,height=0.46\textwidth]{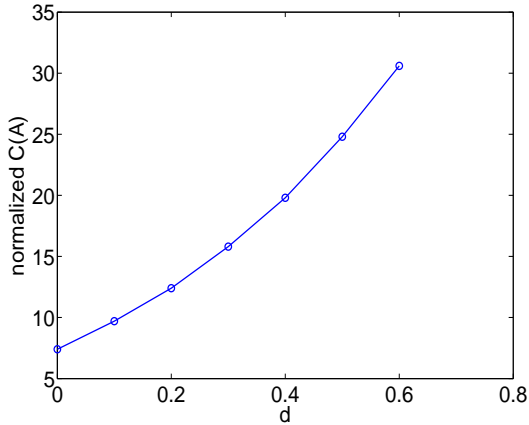}
\caption{Plot of of the total network communicability $C(A)$ for small world graphs 
with increasing probability $d$ of a shortcut. The computed values were averaged over 20 instances.}
\label{fig:norm_C(A)}
\end{figure}

\subsection{Discussion of test results using synthetic data}

The results reported so far can be explained as follows.
In a (regular) ring-shaped network, no node is more central than
the other nodes and no reasonable centrality measure would be
able to assign a (strict) ranking of the nodes.  
In a small world network obtained by perturbing a regular
ring-shaped network, all the nodes have {\em approximately} the same 
importance, with the nodes with extra links (``shortcuts'') being 
slightly more important than the others.  When $d$ is small, these 
shortcuts matter 
more, but the subgraph centrality scores and the total 
communicability scores do not have a large range.  Due to this, 
the change in the scores due to moving from the subgraph 
centrality measure to the total communicability measure can 
have a high impact on node rankings.  This leads to a low correlation and a 
relatively large intersection 
distance  
between the two rankings.  When $d$ gets very large, the shortcuts 
matter less and cause less perturbations between the two sets of rankings. 
By contrast, in a scale-free preferential attachment network 
both the subgraph centrality scores and total 
communicability scores are spread out over a large range, 
even for small $d$, and adding the corresponding
off-diagonal row sums to the
diagonal entries does not change the rankings as much.

\subsection{Real data}
\label{sec:realdata}
Next, we study correlations between the two ranking methods 
using various networks corresponding to real data.  The networks 
in this section come from a variety of sources.  The Zachary Karate 
Club network is a classic example in network analysis \cite{karate}. 
The Intravenous 
Drug User and Yeast PPI networks were provided to us by Prof.~Ernesto 
Estrada. The Yeast PPI network has 440 ones on the diagonal due to the 
self-interactions of certain proteins.  
The remainder of the networks can be found in the University 
of Florida Sparse Matrix Collection \cite{UFsparse} under different
\lq\lq groups\rq\rq. 
The Erd\"os networks are from the Pajek group. They represent various 
subnetworks of the Erd\"os collaboration network. The ca-GrQc and 
ca-HepTh from the SNAP group are collaboration networks for the 
arXiv General Relativity and High Energy Physics Theory subsections, 
respectively. The as-735 network, also from the SNAP group, contains 
the communication network of a group of Autonomous Systems (AS) measured 
over 735 days between November 8, 1997 
and January 2, 2000.  Communication occurs when routers from two Autonomous 
Systems exchange information.  The Minnesota network from the Gleich 
group represents the Minnesota road network. 
The order $n$ and number of nonzeros $nnz$ of the corresponding adjacency 
matrices are given in Table
\ref{tbl:real_exp_compare}.  These networks exhibit a wide variety 
of structural properties and together constitute a 
rather heterogeneous sample of real-world networks. 
All networks except the Yeast PPI network are simple and all are undirected.

\begin{table}
\centering
\caption{Comparison of rankings based on the diagonal and row sum of 
$e^A$ for various real-world networks.}
\begin{tabular}{|c|c|c|c|c|c|c|c|}
\hline
Graph & $n$ & $nnz$ & $\lambda_1$ & $\lambda_2$ & $\mathpzc{cc}$ & $\mathpzc{cc}_{10}$ & $\mathpzc{cc}_{1}$ \\
 \hline
 \hline
Zachary Karate Club & 34 & 156 & 6.726 & 4.977 & 0.420 & -- & 1 \\
 \hline
 Drug User & 616 & 4024 & 18.010 & 14.234 & 0.083 & 0.976 & 1 \\
 \hline
 Yeast PPI & 2224 & 13218 & 19.486 & 16.134 & 0.108 & -- & 1 \\
 \hline
Pajek/Erdos971 & 472 & 2628 & 16.710 & 10.199 & 0.523 & 1 & 1 \\
\hline
Pajek/Erdos972 & 5488 & 14170 & 14.448 & 11.886 & 0.122 & -- & -- \\
\hline
Pajek/Erdos982 & 5822 & 14750 & 14.819 & 12.005 & 0.128 & -- & -- \\
\hline
Pajek/Erdos992 & 6100 & 15030 & 15.131 & 12.092 & 0.143 & -- & -- \\
\hline
SNAP/ca-GrQc & 5242 & 28980 & 45.617 & 38.122 & 0.021 & -- & 0.995 \\
\hline
SNAP/ca-HepTh & 9877 & 51971 & 31.035 & 23.004 & 0.007 & -- & -- \\
\hline
SNAP/as-735 & 7716 & 26467 & 46.893 & 27.823 & 0.904 & 0.771 & 1  \\
\hline
Gleich/Minnesota & 2642 & 6606 & 3.2324 & 3.2319 & 0.087 & -- & -- \\
 \hline
\end{tabular}
\label{tbl:real_exp_compare}
\end{table}

\begin{table}
\centering
\caption{Intersection distance comparison of rankings based on the diagonal and row sum of 
$e^A$ for various real-world networks.} 
\begin{tabular}{|c|c|c|c|c|c|c|c|}
\hline
Graph & isim & ${\rm isim}_{10\%}$ & ${\rm isim}_{1\%}$ \\
 \hline
 \hline
Zachary Karate Club & 0.044 & 0.111 & 0 \\
 \hline
 Drug User & 0.102 & 0.002 & 0 \\
 \hline
 Yeast PPI & 0.025 & 0.056 & 0 \\
 \hline
Pajek/Erdos971 & 0.004 & 0 & 0 \\
\hline
Pajek/Erdos972 & 0.081 & 0.075 & 0.047 \\
\hline
Pajek/Erdos982 & 0.079 & 0.065 & 0.044  \\
\hline
Pajek/Erdos992 & 0.077 & 0.055 & 0.034 \\
\hline
SNAP/ca-GrQc & 0.043 & 0.091 & 5.49e-4  \\
\hline
SNAP/ca-HepTh & 0.142 & 0.319 & 0.134  \\
\hline
SNAP/as-735 & 1.81e-4 & 0.001 & 0  \\
\hline
Gleich/Minnesota & 0.096 & 0.341 & 0.709 \\
 \hline
\end{tabular}
\label{tbl:real_exp_compare_isim}
\end{table}

Table \ref{tbl:real_exp_compare} reports the correlation coefficients 
between the two sets of rankings for all the nodes, the top 10\% of 
the nodes and the top 1\% of the nodes (limited to the cases where the 
two methods rank the same nodes in the top 10\% and top 1\%), 
as well as the value of the
two largest eigenvalues $\lambda_1$ and $\lambda_2$ of the adjacency 
matrix. A ``--" in the table signifies that different lists of top nodes
where produced under the two rankings, hence correlation coefficients
could not be computed in such cases.
Table \ref{tbl:real_exp_compare_isim} reports the intersection distances 
between the two sets of rankings for all, for the top 10\%, and for the top 1\% of the nodes.
Table \ref{tbl:real_tnc} reports the normalized Estrada index 
and normalized total network connectivity for each of the networks. 
For the Zachary Karate Club, which
only has 34 nodes, $\mathpzc{cc}_1 = 1$ and ${\rm isim}_{1\%}=0$ indicate that the top two ranked 
nodes under the two rankings are the same. The top node is node 34, which 
corresponds to the president of the karate club, and the second is node 1, 
which corresponds to the instructor.  These were the two most influential 
members of the club and fought with each other to the point that 
eventually the club split into two factions aligned around each of them 
\cite{karate}. 

The results indicate that there is a good deal of variation 
between the correlation coefficients for these networks.  The correlation 
coefficient between the rankings of all the nodes ranges from a low of 
$0.007$ for the SNAP/ca-HepTh network to a high of $0.904$ for the  
SNAP/as-735 network. Even for
networks that come from similar datasets, the correlation 
coefficients can be very different.  For example, the networks 
in the Pajek group are all subsets of the Erd\"os collaboration 
network, but correlations between the two sets of rankings 
range between 0.122 for the Erdos972 network and 0.583 for 
the Erdos971 network. 

For most of the networks, the correlation coefficient (when defined) increases 
when only the top 1\% of nodes are considered ($\mathpzc{cc}_{1}$), 
sometimes greatly. Five of the networks (Zachary 
Karate Club, Drug User, Yeast PPI, Pajek/Erdos971, and SNAP/as-735) produce the 
exact same rankings on the top 1\% of nodes.  Another network 
(SNAP/ca-GrQc) has a correlation coefficient greater 
than 0.9 on the top 1\% of nodes. 

\begin{table}
\centering
\caption{Comparison of the normalized Estrada index $EE(A)/n$, 
the normalized total network connectivity $C(A)/n$, and 
$e^{\|A\|_2}$ ($=e^{\lambda_1}$) for various real-world networks.}
\begin{tabular}{|c|c|c|c|}
\hline
Graph & normalized $EE(A)$ & normalized $C(A)$ & $e^{\|A\|_2}$  \\
 \hline
 \hline
Zachary Karate Club & 30.62 & 608.79 & 833.81 \\
 \hline
 Drug User & 1.12e05 & 1.15e07 & 6.63e07 \\
 \hline
 Yeast PPI & 1.37e05 & 3.97e07 & 2.90e08 \\
 \hline
Pajek/Erdos971 & 3.84e04 & 4.20e06 & 1.81e07 \\
\hline
Pajek/Erdos972 & 408.23 & 1.53e05 & 1.88e06 \\
\hline
Pajek/Erdos982 & 538.58 & 2.07e05 & 2.73e06 \\
\hline
Pajek/Erdos992 & 678.87 & 2.50e05 & 3.73e06 \\
\hline
SNAP/ca-GrQc & 1.24e16 & 8.80e17 & 6.47e19 \\
\hline
SNAP/ca-HepTh & 3.05e09 & 1.06e11 & 3.01e13 \\
\hline
SNAP/as-735 & 3.00e16 & 3.64e19 & 2.32e20 \\
\hline
Gleich/Minnesota & 2.86 & 14.13 & 35.34 \\
 \hline
\end{tabular}
\label{tbl:real_tnc}
\end{table}

The intersection distance values behave in a similar way, although there is not as much variation in the values.  Among all the nodes, the smallest intersection distance is 1.81e-4 for the as-735 network and the largest is 0.142 for the ca-HepTh network.  These networks also had the largest and smallest correlation coefficients, respectively, for the full set of nodes.  For 5 of the 11 networks examined, the intersection distance value decreases when only the top 10\% of nodes are considered and for all
cases except for the Minnesota road network, it decreases when only the top 1\% of 
nodes are considered. 

It is interesting to note that the similarity between the two ranking 
methods is very different on the ca-GrQc and the ca-HepTh networks.  
The two networks are both arXiv collaboration networks from subsections 
of physics so, intuitively, one would assume that they behaved similarly.  
However, the two rankings are very different on the ca-HepTh network and 
are highly correlated on the ca-GrQc network.  The ca-GrQc network has 
a spectral gap of approximately 7.5 while the spectral gap of ca-HepTh 
is approximately 8, only slightly larger. The relative spectral gaps are
also comparable.  Thus, it is clear that the spectral gap alone cannot
be used to differentiate between the two ranking methods. It appears
that while the two networks are both physics collaboration networks, 
there are significant structural differences between the two groups
which cause the two ranking systems to behave very differently.  
Some insight can be gleaned by looking at the degree distributions of the two 
networks.  Although the ca-HepTh network is almost twice as large as the 
ca-GrQc, the maximum degree on the network is only 65 while the maximum degree on 
the ca-GrQc network is 81.  See Fig.~\ref{fig:collab_degreedist} for the 
degree distributions of the two networks. Additionally, the total 
communicability scores achieved by nodes in the ca-GrQc network range 
from 2.7 to 8.5e19 (the subgraph centrality scores range from 1.5 to 1.6e18).  
In contrast, even with many more nodes, the total communicability scores of 
the ca-HepTh network have a smaller range, from 2.7 to 3.2e13 (the subgraph centrality 
scores range from 1.5 to 9.7e11). It appears that
the wider range of scores in the ca-GrQc 
network helps to prevent rankings from being changed when the scores are perturbed 
by the addition of off-diagonal communicabilities.   
This can be observed when looking at the intersection distances 
between the two sets of rankings on the networks, which are plotted in Fig.~\ref{fig:collab_isim}.  Overall, the intersection distances are much lower for the ca-GrQc network than for 
the ca-HepTh network.  Additionally, for $k\leq34$, ${\rm isim}_k$(ca-GrQc)$=0$, indicating that the first 34 nodes are ranked exactly the same. In contrast,  ${\rm isim}_k$(ca-HepTh)$=0$ only for $k\leq 5$, after which there is a large jump in the intersection distances.

\begin{figure}[t!]
\centering
\includegraphics[width=0.46\textwidth,height=0.46\textwidth]{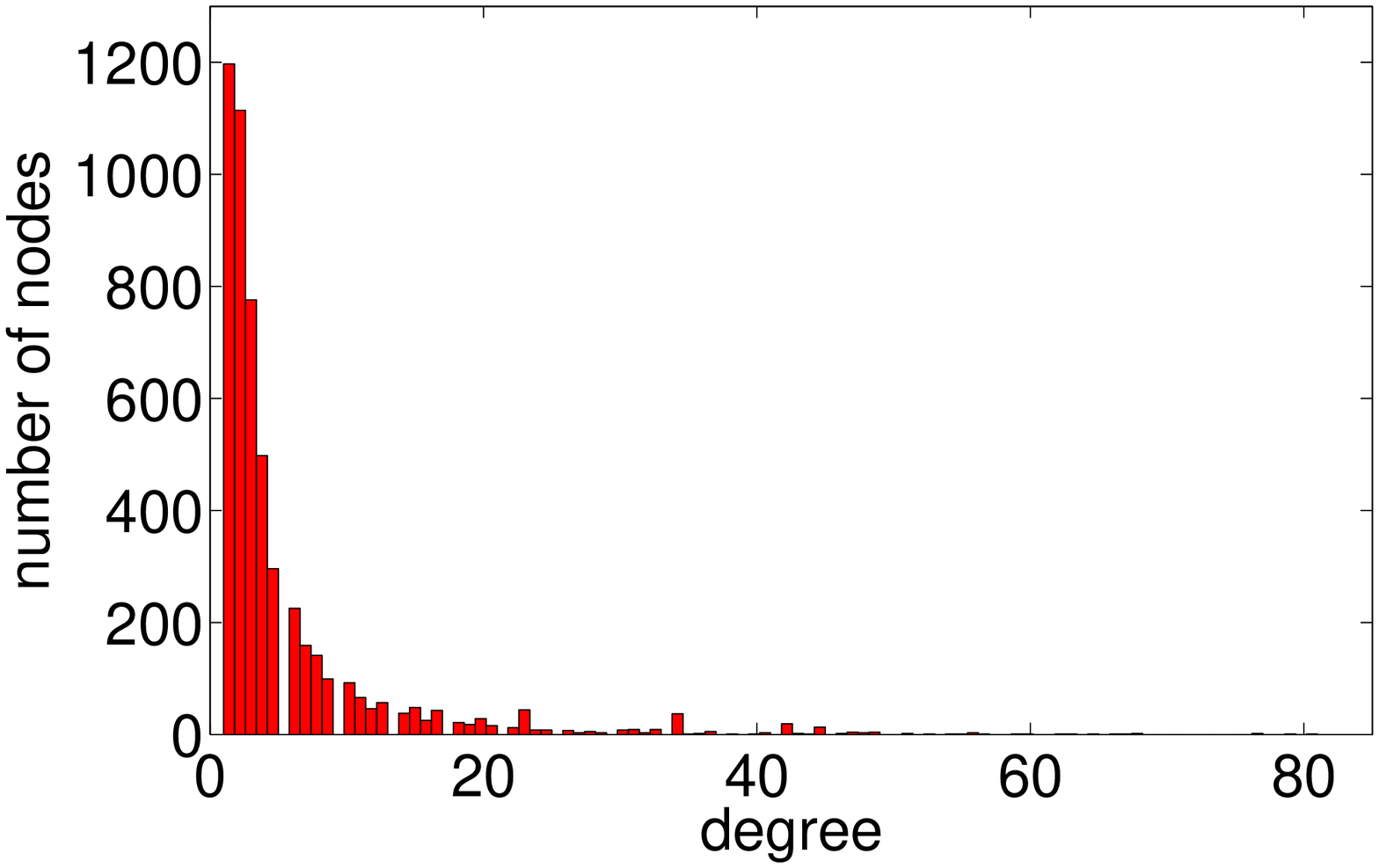}
\hspace{0.3in}
\includegraphics[width=0.46\textwidth,height=0.46\textwidth]{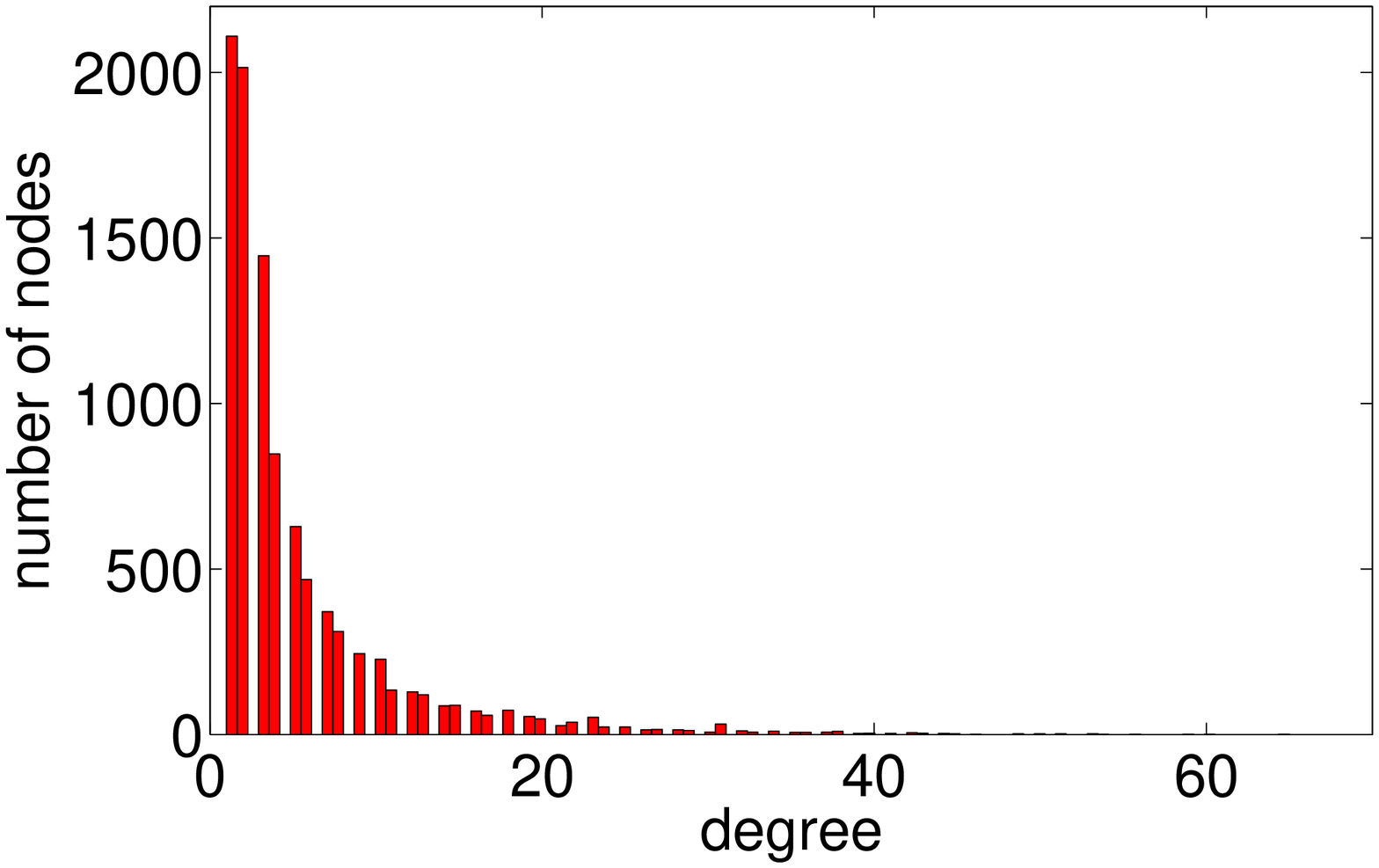}
\caption{The degree distributions of the ca-GrQc (left) and the ca-HepTh (right) collaboration networks.}
\label{fig:collab_degreedist}
\end{figure}

\begin{figure}[t!]
\centering
\includegraphics[width=0.46\textwidth,height=0.46\textwidth]{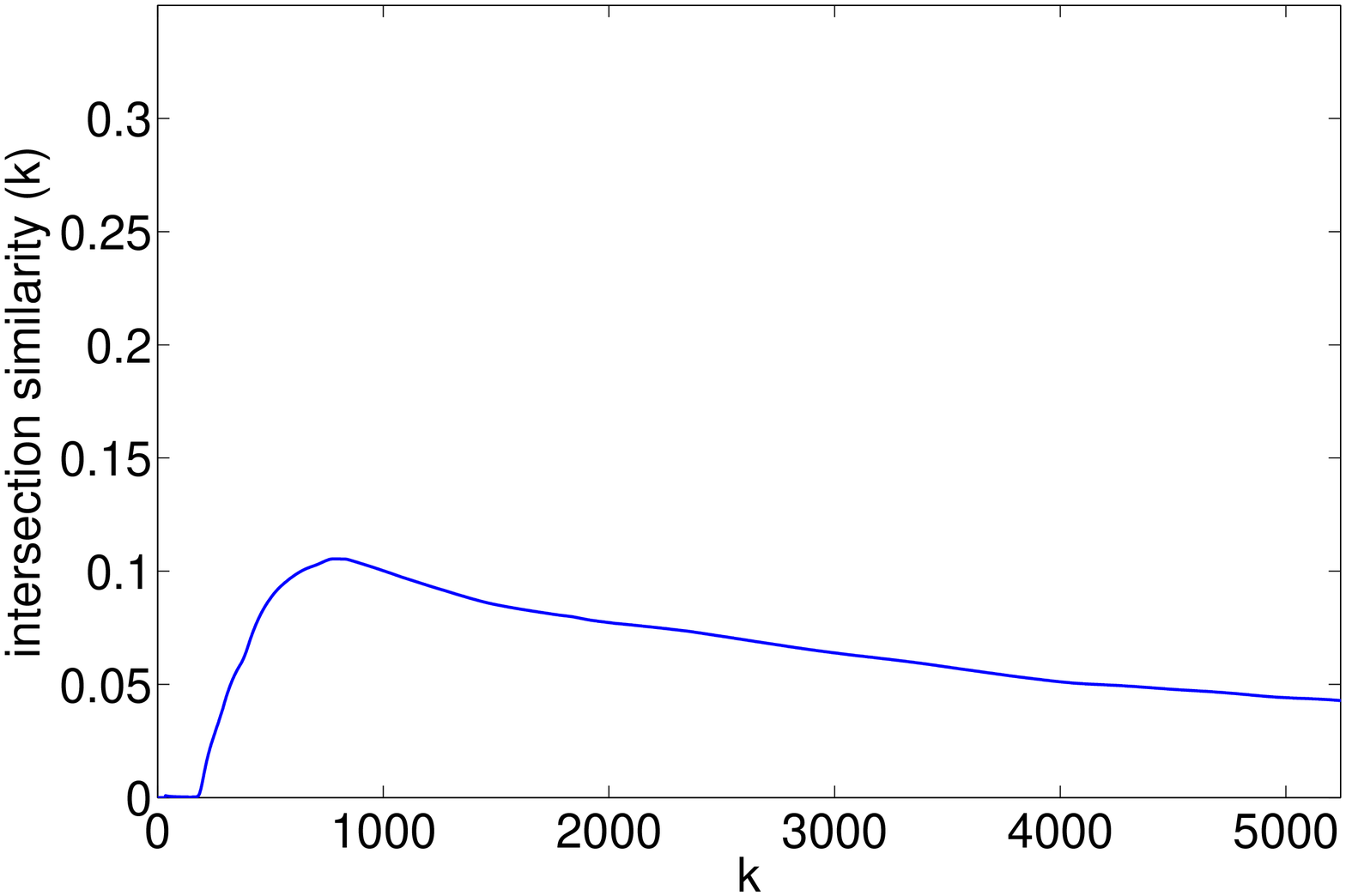}
\hspace{0.3in}
\includegraphics[width=0.46\textwidth,height=0.46\textwidth]{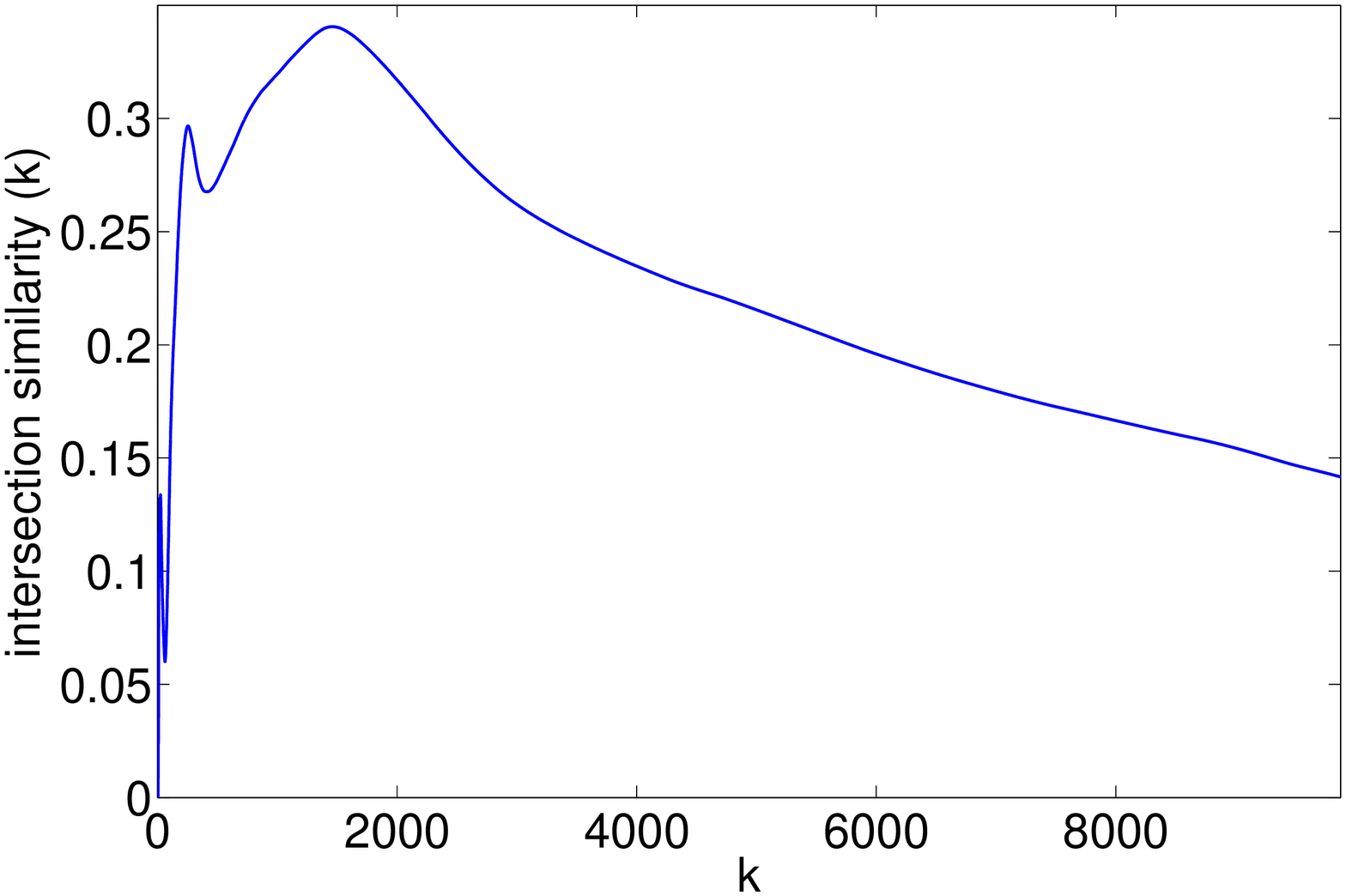}
\caption{The intersection distance values (${\rm isim}_k$) of the ca-GrQc (left) and the ca-HepTh (right) collaboration networks.}
\label{fig:collab_isim}
\end{figure}

Similar behavior can be observed on the various instances of the Erd\"os 
collaboration network. Erdos971, which is very small, shows a high 
correlation between the two rankings; indeed, the rankings of the 
top 10\% of nodes are exactly the same.  On the other instances 
of the collaboration network, however, the rankings are 
somewhat different, as can be seen from the relatively low values of
the correlation
coeffcients. 
The intersection distance values, while not very high, are 
somewhat higher than for most other networks.
The maximum 
subgraph centrality and total communicability scores of the Erdos972 
network are the smallest of any of the Erd\"os collaboration subgraphs.  
The maximum subgraph centrality score is 1.18e05 and the maximum total 
centrality score is 9.20e06.  By comparison, on the (much smaller) 
Erdos971 network, the maximum subgraph centrality score is 1.11e06.  
On the Erdos982 network, the maximum subgraph centrality score is 
1.71e05 and on the Erdos992 network it is 2.47e05.  Although the 
top 5 nodes of the Erdos972 network are exactly the same under the 
two ranking schemes, the relatively narrow range of possible scores  
means that the addition of off-diagonal values to the diagonal ones
perturbs the rankings of the other nodes so much as to 
result in 
a relatively high value of the intersection distance among the top 1\% of nodes.

As before, the spectral gap for these networks does not give
much insight into the behavior of the two ranking schemes, unless
it is really large; the largest spectral gap for this set of
test problems occur for SNAP/as-735, and indeed here we
observe a strong correlation and a small intersection distance between the two metrics. 
Conversely, for the (planar, fairly regular) Gleich/Minnesota 
network, the spectral gap is smallest and not surprisingly
the correlation is very weak and the intersection distance 
for the top 1\% of the nodes, ${\rm isim}_{1\%}$, is very 
high at 0.709. 

When examining the (normalized) total network connectivities of the 
various networks (see Table \ref{tbl:real_tnc}), it can be seen 
that the ease of information sharing across the networks varies 
widely.  Some networks, such as the collaboration networks 
ca-HepTh and ca-GrQc, have a high normalized $C(A)$ (8.80e17 and 
1.06e11, respectively). The value is even higher for the SNAP/as-735
router network ($C(A)/n = $3.64e19). 
 The Minnesota road network, on the other hand, 
has a normalized $C(A)$ of only 14.13, indicating that the network 
is relatively poorly connected, as one would expect in a graph
characterized by wide diameter, small bandwidth and high locality.

\subsection{Identification of essential proteins in PPI network of yeast}

One important application of node centrality measures is to 
rank nodes in protein-protein interaction networks (PPIs) in 
an attempt to determine which proteins are essential, in the
sense that their removal would result in the death of the cell.  
The goal of such rankings is for as many of the top-ranked 
nodes as possible to correspond to essential proteins.  In 
\cite{Proteomics}, various centrality measures were tested 
on their ability to identify essential proteins in the 
Yeast PPI network.  It was shown that, among the centrality 
measures tested, subgraph centrality identified the highest 
percentage of essential proteins ranked in the top 30 nodes, 
identifying 18 essential proteins (in \cite{Proteomics}, subgraph 
centrality was said to identify 19 essential proteins, but 
this was later corrected \cite{EstradaPC}).
When total communicability is used 
instead, the top 30 nodes are the same, so the same percentage 
of essential proteins are identified.  
The intersection distances between the two sets of rankings 
are displayed in Fig.~\ref{fig:yeast_isim}.  Here, it can be seen that the intersection 
distances are small for approximately the top 50 nodes, then they begin to rise.  
The two rankings are least similar for nodes ranked 200-500, then their similarity 
increases again.  As already noted, total 
communicability rankings can be calculated much more quickly 
than subgraph centrality rankings (see also section \ref{sec:approx}).  
Although there are currently methodologies which do better in 
protein ranking (see \cite{Proteome} for example), our
findings suggest
that total communicability does provide valuable 
information about the relative importance of nodes in the network. 

\begin{figure}[t!]
\centering
\includegraphics[width=0.6\textwidth,height=0.46\textwidth]{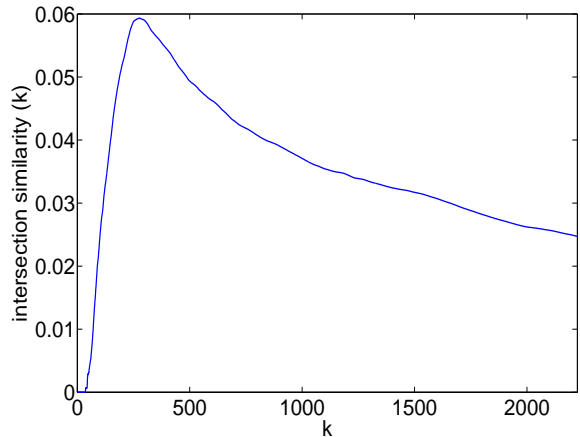}
\caption{The intersection distance values (${\rm isim}_k$) of the Yeast PPI network.}
\label{fig:yeast_isim}
\end{figure}

\subsection{Further discussion of test results using real networks}

The results just described indicate that in general the two
centrality measures can produce significantly different 
rankings, even when one restricts the attention to the
top 1\% of nodes, and even for networks belonging to the
same \lq\lq family\rq\rq.  As in the case of synthetic
networks, a wider range of values in the two sets of
centralities leads to stronger correlations between the
corresponding rankings than in the case of a narrow range.

Two extreme cases are represented by the 
SNAP/as-735 and Gleich/Minnesota data sets. The first one exhibits
a large value of the spectral gap, and thus (as expected)    
a strong correlation between the two rankings; the second
one has tiny spectral gap and results
in very weakly correlated rankings.
For networks that fall somewhere in between these two
extremes, the observed correlation coefficients can vary
significantly. The subgraph centrality scores measure how 
``well-connected'' a node is in the network as a whole while 
the communicability score between nodes $i$ and $j$ measures how 
well information travels between node $i$ and node $j$. Thus, 
the total communicability of node $i$ is a measure of how well 
information travels between node $i$ and any node in the network
(node $i$ itself included).  
Although these two measures are closely related, they are not 
quite the same.  This observation suggests that the two 
centrality measures
reflect somewhat different structural properties of the
networks. Thus, they should be applied in concert rather than
in alternative of one another, unless computational considerations
dictate otherwise.

\section{Computational aspects}
\label{sec:approx}
There are various methods available to compute (or approximate) the 
matrix exponential.  One of the most used schemes (which is implemented 
in Matlab as the {\tt expm} function) is based on Pad\'e approximations 
combined with scaling and squaring \cite{higham,highambook}.  For 
a generic $n \times n$ matrix, this requires $\mathcal{O}(n^3)$ arithmetic 
operations and $\mathcal{O}(n^2)$ storage. The prefactor multiplying
$n^3$ in the arithmetic complexity can vary widely depending on the
sparsity and structural properties of $A$.

Once the matrix exponential 
is computed, both the subgraph centrality and the total 
communicability rankings are readily obtained.
However, to compute the subgraph centrality rankings, we do not need 
the complete matrix exponential, we only need the diagonal entries 
of $e^A$.  Methods for efficiently estimating individual 
entries of matrix functions have been developed by Golub, Meurant, 
and others \cite{golubmeurantbook,BenziGolub} and these methods have 
previously been applied to network analysis \cite{BenziBoito,BEK}.  
They are based on Gaussian quadrature and the Lanczos algorithm,
and they have been implemented in 
the Matlab toolbox {\tt mmq}  \cite{mmq}. The cost per node of estimating 
the subgraph centrality is typically $\mathcal{O}(n)$, giving a total cost of 
approximately $\mathcal{O}(n^2)$ for estimating the subgraph centrality for every 
node and computing the subgraph centrality rankings.  However, the 
coefficient of the $\mathcal{O}(n)$ estimate can be quite large.  
Additionally, the {\tt mmq} toolbox-based implementation for calculating 
subgraph centrality that we use here has not been optimized, unlike
the built-in Matlab function {\tt expm}. We mention in passing that 
methods for quickly determining the top $k$ nodes and only calculating 
the exact rankings on this subset have also been developed \cite{BEK,fmrr12}.

The individual entries of the matrix exponential are not necessary
for computing the total communicability rankings; only the 
row sums of $e^A$ are necessary. An efficient algorithm
for evaluating $f(A){\bf v}$ using 
a restarted Krylov method has recently been presented in \cite{Krylov1,Krylov2}.  
In this approach, the basic operation is represented by matrix-vector
products with $A$.
This method has been implemented in the Matlab toolbox {\tt  funm\_kryl} 
by Stefan G\"uttel \cite{funmkryl}. We apply this algorithm with $f(A) = e^A$
and ${\bf v} = {\bf 1}$. Clearly, the same algorithm can be used to
rapidly compute $C(A) = {\bf 1}^T e^A {\bf 1}$.
For many network of practical interest, the cost is typically
$\mathcal{O}(n)$, although the prefactor can vary considerably
for different types of networks.

\begin{table}
\centering
\caption{Timings (in seconds) to compute
centrality rankings based on the diagonal and row sum of 
$e^A$ for various test problems using different methods.} 
\begin{tabular}{|c|c|c|c|}
\hline
Graph & {\tt expm} & {\tt mmq} & {\tt funm\_kryl} \\
 \hline
 \hline
Zachary Karate Club & 0.062 & 0.138 &  0.120  \\
 \hline
 Drug User & 0.746 & 2.416 &  0.363 \\
 \hline
 Yeast PPI & 47.794 & 9.341 & 0.402 \\
 \hline
Pajek/Erdos971 & 0.542 & 2.447 & 0.317  \\
\hline
Pajek/Erdos972 & 579.214 & 35.674 & 0.410  \\
\hline
Pajek/Erdos982 & 612.920 & 39.242 & 0.393  \\
\hline
Pajek/Erdos992 & 656.270 & 53.019 & 0.325  \\
\hline
SNAP/ca-GrQc & 281.814 & 23.603 & 0.465  \\
\hline
SNAP/ca-HepTh & 2710.802 & 58.377 & 0.435 \\
\hline
SNAP/as-735 & 2041.439 & 75.619 & 0.498  \\
\hline
Gleich/Minnesota & 1.956 & 10.955 & 0.329 \\
 \hline
\end{tabular}
\label{tbl:real_exp_timings}
\end{table}

Table \ref{tbl:real_exp_timings} lists the timings for 
calculating the matrix exponential directly using {\tt expm}, 
estimating the subgraph centralities using the {\tt mmq} 
toolbox (with 5 iterations of the Lanczos algorithm per node),  
and estimating the total communicabilities using the 
{\tt funm\_kryl} toolbox to estimate $e^A{\bf 1}$ (using a 
very stringent stopping tolerance of 1e-16). 
These computations have been performed using Matlab Version 7.9.0 (R2009b) on 
a 2.4 GHZ Intel Core i5 processor with 4 GB of RAM.  
In general, the timings with {\tt expm} increase for
increasing number of nodes, but structural properties 
of the underlying graph, like the network diameter,
can have a very significant impact on
the computing times. For example, the
yeast PPI network and the Minnesota 
road network have approximately the same number $n$ of nodes (2224 and 
2642, respectively), yet computing the matrix exponential 
for the yeast network takes almost 25 times longer than for 
the Minnesota road network. 
This appears to be due to the fact that the yeast 
network has a much smaller diameter 
than the Minnesota network, therefore the
powers $A^k$ of the adjacency matrix fill up much more
quickly. Since the algorithm implemented in ${\tt expm}$
involves solving linear systems with polynomials
in $A$ as coefficient matrices, the execution time for
sparse matrices with small diameter tends to be much
higher than for matrices exhibiting a high degree of locality.  

For the majority of the networks tested, using the {\tt mmq} 
toolbox to estimate subgraph centrality was faster than using 
{\tt expm}, frequently by far.  
The exceptions (Zachary Karate Club, Drug User, 
Erdos971, and Minnesota) were the networks with a small number 
of nodes and/or a high diameter.  

The computation of the total communicabilities using 
{\tt funm\_kryl} was by far the fastest method for all networks tested, 
with the only exception of the tiny Zachary Karate Club network. In principle,
this is a clear
advantage of total communicability over subgraph
centrality. However, as we saw, the two methods often result
in rather different rankings, therefore we cannot simply
replace subgraph centrality with total communicability.

\subsection{A large-scale example}\label{sec:large}
In addition to the test results discussed above, we performed 
tests with the digraph of Wikipedia (as of June 6, 2011),
where nodes correspond to 
entries and directed links to hyperlinks from one entry to another.
In this case, the entries of $e^A{\bf 1}$ provide a ranking
of the hubs in the networks, see \cite{BEK}. 
This graph contains 4,189,503 nodes and 67,197,636 links, and it
is prohibitively large for centrality measures based on estimating
the diagonals of the matrix exponential.
For this reason, we limit ourselves to computations using the
{\tt funm\_kryl} toolbox to estimate the row sum vector $e^A{\bf 1}$. 
The restart parameter was set to 10 and we allowed a maximum of
50 restarts. The run time to obtain the rankings on a parallel system comprising  
24 Intel(R) Xeon(R) E5-2630 2.30GHz CPU(s) was 216.7 seconds.
This shows that centrality calculations using total communicability 
are quite feasible even for large networks.

\section{Resolvent-based centrality measures}
\label{sec:resolvent}
There are matrix functions other than the matrix exponential 
that may be used to calculate subgraph centrality and subgraph 
communicability scores.  The most common of these is the matrix resolvent
\begin{equation}
	(I-\alpha A)^{-1} = I + \alpha A +\alpha ^2A^2+ \cdots 
+ \alpha ^kA^k + \cdots = \sum_{k=0}^\infty \alpha ^kA^k\,,
\end{equation}
where $0 < \alpha < \frac{1}{\rho (A)}$, with $\rho(A)$ the
spectral radius of $A$.  This was first 
used by Katz \cite{Katz} in the 1950s and has been used in various 
forms since then \cite{Bonacich,bonchietal,Brandes,EHB11,NetworkProp,IRSurvey,Pagerank}.  
The bounds on $\alpha$ ensure that $I-\alpha A$ is invertible and that the 
geometric series converges to the inverse.  Additionally, the 
inverse is nonnegative; indeed, $I-\alpha A$ is a nonsingular $M$-matrix.
Note that if $A$ is the adjacency matrix of an undirected
network, $\rho(A) = \lambda_{\max}(A) =  \|A\|_2$. 
Since the spectral radius of a nonnegative matrix always
satisfies $\rho(A) \ge \min_i \sum_{j=1}^n a_{ij}$, it follows that
for a connected undirected graph $\alpha$ must be less than 1. 

Like the matrix exponential, $[(I-\alpha A)^{-1}]_{ii}$ counts the number 
of closed walks centered at node $i$ and $\sum_{j=1}^n [(I-\alpha A)^{-1}]_{ij}$ 
counts all walks between node $i$ and all other nodes in the network.  
In this case, however, a walk of length $k$ is penalized by a factor 
of $\alpha^k$.  One drawback of the use of the matrix resolvent in 
determining centrality rankings is the need to choose the value of $\alpha$; 
also, different values of $\alpha$ can lead to different rankings. 
For the purposes of the experiments below, we select $\alpha = \frac{0.85}{\lambda_{\max}(A)}$ 
(similar to the choice of parameter in PageRank \cite{Pagerank}). 

Resolvent-based total network communicability can also 
be evaluated.  As when using the matrix exponential (cf.~section 
\ref{sec:total_network}), the resolvent-based total network communicability 
is an upper bound for the resolvent-based Estrada index. In the following,
$C_r(A)=\sum_{i=1}^n\sum_{j=1}^n \left [(I-\alpha A)^{-1}\right ]_{ij}$
denotes the resolvent-based total communicability of a network.
The following Proposition can be easily proved along the
same lines as Proposition \ref{prop1}.

\begin{prop} Let $A$ be the adjacency matrix of a simple, 
undirected network on $n$ vertices.  Then for any 
$0 < \alpha < \frac{1}{\|A\|_2}$, $$EE_r(A) := {\rm Tr}\left[(I-\alpha A)^{-1}\right] 
\leq C_r(A) \leq \frac{n}{1-\alpha \|A\|_2}.$$
For an undirected network, $\lambda_{\max}(A) = \lambda_1$ 
can replace $\|A\|_2$ in 
the upper bound above.
 \end{prop}  

\begin{table}[t!]
  \centering
  \caption{Comparison using correlation coefficients of rankings based on the diagonal entries 
and row sums of $(I-\alpha A)^{-1}$ for 1000-node 
scale-free networks of various parameters built using the 
{\tt pref} function in the CONTEST Matlab toolbox.  For each 
instance, the results are measured for $\alpha =\frac{0.85}{\lambda_{\max}(A)}$.  
The values reported are the averages over 20 matrices with the same parameters.}
  \subfloat[][]{\begin{tabular}{|c|c|}
\hline
$d$  & $\mathpzc{cc}$ \\
 \hline
 \hline
 1 & 0.292 \\
 \hline
 2 & 0.370 \\
\hline
 3 & 0.442 \\
\hline
4 & 0.486 \\
\hline
5 & 0.536 \\
\hline
6 & 0.583 \\
\hline
7 & 0.607 \\
\hline
8 & 0.638 \\
\hline
9 & 0.667 \\
 \hline
\end{tabular}}%
  \qquad
  \subfloat[][]{\begin{tabular}{|c|c|}
\hline
$d$  & $\mathpzc{cc}$ \\
 \hline
 \hline
10 & 0.691 \\
  \hline
 20 & 0.840 \\
\hline
 30 & 0.890 \\
\hline
40 & 0.917 \\
\hline
50 & 0.933\\
\hline
60 & 0.942 \\
\hline
70 & 0.949 \\
\hline
 80 & 0.954 \\
\hline
90 & 0.958 \\
\hline
100 & 0.962 \\
\hline
\end{tabular}}%
  \qquad
  \subfloat[][]{\begin{tabular}{|c|c|}
\hline
$d$  & $\mathpzc{cc}$ \\
 \hline
\hline
110 & 0.964 \\
\hline
120 & 0.965 \\
\hline
130 & 0.968 \\
\hline
 140 & 0.970 \\
\hline
150 & 0.971 \\
\hline
160 & 0.973 \\
\hline
170 & 0.973 \\
\hline
 180 & 0.975 \\
\hline
190 & 0.976 \\
\hline
200 & 0.976 \\
\hline
\end{tabular}}%

 \label{tbl:pref_resolvent_compare}%
\end{table}

\begin{table}[t!]
  \centering
  \caption{Intersection distance comparison of rankings based on the diagonal entries 
and row sums of $(I-\alpha A)^{-1}$ for 1000-node 
scale-free networks of various parameters built using the 
{\tt pref} function in the CONTEST Matlab toolbox.  For each 
instance, the results are measured for $\alpha =\frac{0.85}{\lambda_{\max}(A)}$.  
The values reported are the averages over 20 matrices with the same parameters.}
  \subfloat[][]{\begin{tabular}{|c|c|c|}
\hline
$d$  & isim & ${\rm isim}_{10\%}$  \\
 \hline
 \hline
 1 & 0.186 & 0.491 \\
 \hline
 2 & 0.205 & 0.364 \\
\hline
 3 & 0.192 & 0.235 \\
\hline
4 & 0.179 & 0.173  \\
\hline
5 & 0.163 & 0.126 \\
\hline
6 & 0.150 & 0.102 \\
\hline
7 & 0.137 & 0.082 \\
\hline
8 & 0.124 & 0.068 \\
\hline
9 & 0.115 & 0.059 \\
 \hline
\end{tabular}}%
  \qquad
  \subfloat[][]{\begin{tabular}{|c|c|c|}
\hline
$d$  & isim & ${\rm isim}_{10\%}$  \\
 \hline
 \hline
10 & 0.105 & 0.051 \\
  \hline
 20 & 0.055 & 0.020 \\
\hline
 30 & 0.035 & 0.012 \\
\hline
40 & 0.025 & 0.007 \\
\hline
50 & 0.019 & 0.005 \\
\hline
60 & 0.015 & 0.004 \\
\hline
70 & 0.012 & 0.003 \\
\hline
 80 & 0.010 & 0.002 \\
\hline
90 & 0.009 & 0.002 \\
\hline
100 & 0.007 & 0.001  \\
\hline
\end{tabular}}%
  \qquad
  \subfloat[][]{\begin{tabular}{|c|c|c|}
\hline
$d$  & isim & ${\rm isim}_{10\%}$  \\
 \hline
\hline
110 & 0.006 & 0.001 \\
\hline
120 & 0.005 & 7.12e-4 \\
\hline
130 & 0.005 & 6.98e-4 \\
\hline
 140 & 0.004 & 5.74e-4 \\
\hline
150 & 0.004 & 5.62e-4  \\
\hline
160 & 0.003 & 3.69e-4  \\
\hline
170 & 0.003 & 4.25e-4 \\
\hline
 180 & 0.003 & 3.11e-4 \\
\hline
190 & 0.003 & 3.16e-4 \\
\hline
200 & 0.002 & 4.00e-4  \\
\hline
\end{tabular}}%
 \label{tbl:pref_resolvent_compare_isim}%
\end{table}

The resolvent-based subgraph centrality and total 
communicability rankings were compared on the same two sets of synthetic 
networks used for the tests in section \ref{sec:tests}.

Table \ref{tbl:pref_resolvent_compare} lists the average correlation 
coefficient between the subgraph centrality and total 
communicability rankings for  the nodes in networks constructed using the 
preferential attachment model (function {\tt pref} in CONTEST) 
and Table \ref{tbl:pref_resolvent_compare_isim} lists the 
intersection distances for all the nodes and for the top 10\% of the nodes. 
For small values of $d$ ($1 \leq d \leq 3$), 
the correlation coefficients between the two sets of rankings 
using the matrix resolvent are close to those using the matrix 
exponential. 
However, when using the matrix 
exponential the average correlation coefficient was found to be
greater than 0.9 
for all $d \geq 4$, and exactly 1 for all $d \geq 8$.  Using the 
matrix resolvent the correlation coefficient grows 
as $d$ increases, but somewhat more slowly than for the
matrix exponential. The intersection distances are also larger 
for all values of $d$ when the matrix resolvent is used, although they 
also decrease as $d$ increases. 
Moreover, we did not find a single instance where the two methods
produced {\em exactly} the same rankings.

\begin{table}%
  \centering
  \caption{Comparison using correlation coefficients of rankings based on the diagonal entries 
and row sums of $(I-\alpha A)^{-1}$ for 1000-node small world 
networks of various parameters built using the {\tt smallw} 
function with $p=0.1$ in the CONTEST Matlab toolbox.  For 
each instance, the results are measured for 
$\alpha =\frac{0.85}{\lambda_{\max}(A)}$.  The values reported are 
the averages over 20 matrices with the same parameters.}
  \subfloat[][]{\begin{tabular}{|c|c|}
\hline
$d$  & $\mathpzc{cc}$ \\
 \hline
 \hline
 1 & 0.065 \\
 \hline
 2 & 0.023 \\
\hline
 3 & 0.052 \\
\hline
4 & 0.052 \\
\hline
5 & 0.052 \\
\hline
6 & 0.051 \\
\hline
7 & 0.062 \\
\hline
8 & 0.037 \\
\hline
9 & 0.050 \\
 \hline
\end{tabular}}%
  \qquad
  \subfloat[][]{\begin{tabular}{|c|c|}
\hline
$d$  & $\mathpzc{cc}$ \\
 \hline
 \hline
10 & 0.063 \\
  \hline
 20 & 0.078 \\
\hline
 30 & 0.080 \\
\hline
40 & 0.135 \\
\hline
50 & 0.144 \\
\hline
60 & 0.141 \\
\hline
70 & 0.144 \\
\hline
 80 & 0.133 \\
\hline
90 & 0.248 \\
\hline
100 & 0.190 \\
\hline
\end{tabular}}%
  \qquad
  \subfloat[][]{\begin{tabular}{|c|c|}
\hline
$d$  & $\mathpzc{cc}$ \\
 \hline
\hline
110 & 0.294 \\
\hline
120 & 0.246 \\
\hline
130 & 0.275 \\
\hline
140 & 0.311 \\
\hline
150 & 0.312 \\
\hline
160 & 0.321 \\
\hline
170 & 0.301 \\
\hline
180 & 0.293 \\
\hline
190 & 0.354 \\
\hline
200 & 0.300 \\
\hline
\end{tabular}}%

 \label{tbl:smallw_resolvent_compare}%
\end{table}

\begin{table}%
  \centering
  \caption{Intersection distance comparison of rankings based on the diagonal entries 
and row sums of $(I-\alpha A)^{-1}$ for 1000-node small world 
networks of various parameters built using the {\tt smallw} 
function with $p=0.1$ in the CONTEST Matlab toolbox.  For 
each instance, the results are measured for 
$\alpha =\frac{0.85}{\lambda_{\max}(A)}$.  The values reported are 
the averages over 20 matrices with the same parameters.}
  \subfloat[][]{\begin{tabular}{|c|c|c|}
\hline
$d$  & isim & ${\rm isim}_{10\%}$  \\
 \hline
 \hline
 1 & 0.040 & 0.149 \\
 \hline
 2 & 0.070 & 0.189 \\
\hline
 3 & 0.085 & 0.241 \\
\hline
4 & 0.091 & 0.269 \\
\hline
5 & 0.098 & 0.301 \\
\hline
6 & 0.104 & 0.318 \\
\hline
7 & 0.126 & 0.361 \\
\hline
8 & 0.135 & 0.414 \\
\hline
9 & 0.149 & 0.413 \\
 \hline
\end{tabular}}%
  \qquad
  \subfloat[][]{\begin{tabular}{|c|c|c|}
\hline
$d$  & isim & ${\rm isim}_{10\%}$  \\
 \hline
 \hline
10 & 0.156 & 0.435 \\
  \hline
 20 & 0.207 & 0.508 \\
\hline
 30 & 0.198 & 0.517 \\
\hline
40 & 0.204 & 0.571 \\
\hline
50 & 0.207 & 0.621 \\
\hline
60 & 0.191 & 0.588 \\
\hline
70 & 0.181 & 0.582 \\
\hline
 80 & 0.189 & 0.607 \\
\hline
90 & 0.156 & 0.597 \\
\hline
100 & 0.179 & 0.585  \\
\hline
\end{tabular}}%
  \qquad
  \subfloat[][]{\begin{tabular}{|c|c|c|}
\hline
$d$  & isim & ${\rm isim}_{10\%}$  \\
 \hline
\hline
110 & 0.147 & 0.541 \\
\hline
120 & 0.148 & 0.553 \\
\hline
130 & 0.160 & 0.554 \\
\hline
140 & 0.142 & 0.560 \\
\hline
150 & 0.123 & 0.542  \\
\hline
160 & 0.121 & 0.539  \\
\hline
170 & 0.124 & 0.517 \\
\hline
180 & 0.125 & 0.512 \\
\hline
190 & 0.114 & 0.504 \\
\hline
200 & 0.123 & 0.504  \\
\hline
\end{tabular}}%

 \label{tbl:smallw_resolvent_compare_isim}%
\end{table}

For the small world networks, all values $1 \leq d \leq 10$ 
as well as as all multiples of 10 with $20 \leq 10 \leq 200$ 
were tested.  For each $d$, twenty networks were tested. The 
averages of the correlation coefficients between the subgraph 
centrality and total communicability rankings can be found
in Table \ref{tbl:smallw_resolvent_compare} and the average 
intersection distances for both all the nodes and the top 10\% 
of the nodes can be found in Table \ref{tbl:smallw_resolvent_compare_isim}. 
As was the case for
the matrix exponential, 
the two methods (diagonal entries and row sums)
using the matrix resolvent exhibit much weaker correlations 
for this class of networks than for the preferential attachment
networks; indeed, the correlations tend to be even smaller 
for the resolvent than for the exponential.
 For $d=1$, the average correlation is 0.065 
and the average intersection distance was 0.040
using the resolvent, compared to a correlation of 0.177 and 
an intersection distance of 0.015 using the exponential.  
For the values of $d$ tested, the highest average correlation 
coefficient was 0.354, for $d=190$.  
When looking at the intersection distances for other values of $d$, 
the picture is somewhat 
different. Comparing Table \ref{tbl:smallw_resolvent_compare_isim}
with Table \ref{tbl:smallw_exp_compare_isim}, we see that for small $d$ 
the intersection distance between the two ranking schemes tends to be
somewhat higher with the matrix exponential than with the resolvent.
However, as $d$ increases the intersection distance eventually drops 
with the matrix exponential, but not with the resolvent. This is true
both when looking at the ranking of all the nodes and when looking
at only the top 10\%. 

\begin{table}
\centering
\caption{Comparison using correlation coefficients of rankings based on the diagonal entries and 
row sums of  $(I-\alpha A)^{-1}$ with  $\alpha=\frac{0.85}{\lambda_{\max}(A)}$ 
for various real-world networks.}
\begin{tabular}{|c|c|c|c|}
\hline
Graph & $\mathpzc{cc}$ & $\mathpzc{cc}_{10}$ & $\mathpzc{cc}_{1}$ \\
 \hline
 \hline
Zachary Karate Club & 0.589 & 1 & 1  \\
 \hline
 Drug User & 0.189 & -- & --  \\
 \hline
 Yeast PPI & 0.177 & -- & --  \\
 \hline
Pajek/Erdos971 & 0.233 & -- & 1 \\
\hline
Pajek/Erdos972 & 0.215 & -- & --  \\
\hline
Pajek/Erdos982 & 0.207 & -- & --  \\
\hline
Pajek/Erdos992 & 0.197 & -- & ---  \\
\hline
SNAP/ca-GrQc & 0.070 & -- & --  \\
\hline
SNAP/ca-HepTh & 0.072 & -- & -- \\
\hline
SNAP/as-735 & 0.204 & -- & ---  \\
\hline
Gleich/Minnesota & 0.019 & -- & --  \\
 \hline
\end{tabular}
\label{tbl:real_resolvent_compare}
\end{table}

\begin{table}
\centering
\caption{Intersection distance comparison of rankings based on the diagonal entries and 
row sums of  $(I-\alpha A)^{-1}$ with  $\alpha=\frac{0.85}{\lambda_{\max}(A)}$ 
for various real-world networks.} 
\begin{tabular}{|c|c|c|c|c|c|c|c|}
\hline
Graph & isim & ${\rm isim}_{10\%}$ & ${\rm isim}_{1\%}$ \\
 \hline
 \hline
Zachary Karate Club & 0.061 & 0 & 0 \\
 \hline
 Drug User & 0.125 & 0.145 & 0.069 \\
 \hline
 Yeast PPI & 0.204 & 0.363 & 0.187 \\
 \hline
Pajek/Erdos971 & 0.080 & 0.050 & 0 \\
\hline
Pajek/Erdos972 & 0.110 & 0.273 & 0.263 \\
\hline
Pajek/Erdos982 & 0.109 & 0.269 & 0.264  \\
\hline
Pajek/Erdos992 & 0.109 & 0.271 & 0.247 \\
\hline
SNAP/ca-GrQc & 0.047 & 0.122 & 0.033  \\
\hline
SNAP/ca-HepTh & 0.058 & 0.159 & 0.236  \\
\hline
SNAP/as-735 & 0.247 & 0.513 & 0.271  \\
\hline
Gleich/Minnesota & 0.102 & 0.301 & 0.557 \\
 \hline
\end{tabular}
\label{tbl:real_resolvent_compare_isim}
\end{table}

Next, we consider tests with real-world networks.
As shown in Table \ref{tbl:real_resolvent_compare}, 
the correlation coefficients 
between the two ranking systems for the whole set of nodes 
were higher (in a majorityy of cases) using the matrix resolvent than they 
were using the matrix exponential. 
(Again, a ``--" signifies that correlation coefficients could
not be computed due to the fact that the two ranking schemes 
produced different lists of nodes.) Only the Erdos971, as-735, 
and the Minnesota networks had a higher correlation coefficient 
between the two ranking systems under the exponential than under 
the matrix resolvent.  This can be understood when looking at the 
normalized Estrada indexes and total network communicabilities 
in Table \ref{tbl:real_resolvent_tnc}.  The smaller the factor 
$\alpha$, the more it minimizes the contribution of the network 
data from $A$ to the scores produced by the diagonal entries or 
row sums of $(I-\alpha A)^{-1}$.  This can also be seen by 
noticing that as $\alpha\rightarrow0$, $(I-\alpha A)^{-1}$ 
approaches the identity.  In these experiments, $\alpha = \frac{0.85}{\lambda_{\max}(A)}$. 
However, this also means that for the networks tested with a large maximum 
eigenvalue (such and ca-GrQc, ca-HepTh, and as-735) 
$\alpha$ is quite small, causing the resulting subgraph 
centrality scores to be small and, consequently, close together.  
In the case of a network with a small maximum eigenvalue 
(such as the Minnesota network), the effect of $\alpha$ is 
not as pronounced. The compression of the score values 
means that a perturbation of the scores (such as occurs 
when switching from subgraph centrality scores to total 
communicability scores) has a large effect on the node 
rankings, especially for the higher ranked nodes.

\begin{table}
\centering
\caption{Comparison of the normalized resolvent-based Estrada index $EE_r(A)/n$ 
and total network connectivity $C_r(A)/n$ for various real-world networks.  
Here, $f(A) = (I-\alpha A)^{-1}$ with $\alpha = \frac{0.85}{\lambda_{\max}(A)}$.}
\begin{tabular}{|c|c|c|}
\hline
Graph & normalized $EE_r(A)$ & normalized $C_r(A)$  \\
 \hline
 \hline
Zachary Karate Club & 1.21 & 5.13 \\
 \hline
 Drug User & 1.03 & 2.36 \\
 \hline
 Yeast PPI & 1.03 & 2.17  \\
 \hline
Pajek/Erdos971 & 1.03 & 2.44 \\
\hline
Pajek/Erdos972 & 1.01 & 1.70 \\
\hline
Pajek/Erdos982 & 1.01 & 1.66 \\
\hline
Pajek/Erdos992 & 1.01 & 1.60 \\
\hline
SNAP/ca-GrQc & 1.00 & 1.21 \\
\hline
SNAP/ca-HepTh & 1.01 & 1.24  \\
\hline
SNAP/as-735 & 1.00 & 1.86  \\
\hline
Gleich/Minnesota & 1.27 & 3.44 \\
 \hline
\end{tabular}
\label{tbl:real_resolvent_tnc}
\end{table}

When only the top 1\% of nodes were considered, the 
exponential subgraph centrality and exponential total 
communicability rankings were much closer together than their resolvent 
counterparts, where often the top 1\% of nodes were not even the same.  This 
seems to indicate that when using the resolvent, 
the subgraph centrality 
and total communicability tend to rank the less important 
nodes more similarly than they do under the matrix exponential.  
Under the matrix exponential, the two rankings seem to agree 
more closely on the important nodes than they do when using 
the resolvent. 
This can also be seen when looking at the intersection distance, which gives more weight to differences in the top ranked nodes than in the lower ranked nodes.  For all networks except ca-HepTh, the intersection distance between the two rankings is smaller when using the exponential than when using the resolvent.  When looking at the top 1\% of nodes, the intersection distances are also 
smaller (often much smaller) in the case of the exponential, for all except three of the networks.  The exceptions are the Minnesota road network (which has a large intersection distance on the top 1\% of nodes for both the exponential and the resolvent) and the Zachary Karate Club and Erdos971 networks (which have isim$_{1\%}=0$ for both cases). 

Another observation that can be made is that the resolvent-based
total network communicability $C_r(A)$ is unable to discriminate
between highly connected networks and poorly connected ones, in
stark contrast with the exponential-based one. For instance, in
the case of the Minnesota road network $\alpha$ is relatively 
large (since $\lambda_1$ is small for this graph), hence the
off-diagonal contributions to $C_r(A)$ are more significant
than for other networks where $\lambda_1$ is large (thus forcing
a small value of $\alpha$, leading to a resolvent very close
to the identity matrix). Thus, only the exponential-based
total network communicability should be used when comparing 
different networks in terms of ease of communication.

When the identification of essential proteins in the Yeast PPI 
network is considered using resolvent-based total communicability, 
the results are comparable to those using the exponential.  
The resolvent-baed total communicability rankings with 
$\alpha = \frac{0.85}{\lambda_{\max}(A)}$ identified 17 essential 
proteins in the top 30 (as compared to 18 identified by 
exponential subgraph centrality and total communicability).  
The resolvent-based subgraph centrality, however, identified 
19 essential proteins in the top 30, slightly outperforming the other methods. 

Concerning the computational complexity, when dealing with
large networks the use of the conjugate gradient 
method (possibly with some type of preconditioning)
to solve the linear system $(I - \alpha A){\bf x} = {\bf 1}$ is
orders of magnitude faster than trying to estimate the diagonal entries
of $(I - \alpha A)^{-1}$. For certain networks, Chebyshev
semi-iteration can be even faster \cite{BK12}.
 Thus, as was the case for the matrix
exponential, rankings based on total communicability 
(row sums) are a lot cheaper than the rankings based
on subgraph centrality (diagonals). Once again, however,
the two ranking methods in general produce different
rankings, so one should not choose between the two based
solely on computational cost.

\section{Conclusions} \label{sec:concl}
We have examined the use of total communicability 
as a method for ranking the importance of nodes in a network.  
Like the subgraph centrality ranking, the total 
communicability ranking using the matrix exponential counts 
the number of walks starting at a given node, weighing walks 
of length $k$ by a penalization factor of $\frac{1}{k!}$.  
However, instead of only counting closed walks, it counts 
all walks between the given node and every node in the network.  
If the matrix resolvent is used, the weight on the walks 
becomes $\alpha ^k$ for a chosen parameter $\alpha$ in a certain range. 
There are various classes of graphs on which it can 
be shown that the two exponential-based rankings are always 
identical or in very good agreement; for instance, certain
types of simple regular graphs and Erd\"os--Renyi random
graphs with large spectral gap.
However, as is well known, these classes are not realistic 
models of real-world complex networks.

The two sets of rankings (total communicability 
and subgraph centrality) have been used to rank the nodes 
of networks corresponding to both real and synthetic data 
sets. The synthetic data sets were constructed using the 
preferential attachment (Barab\'asi--Albert) 
and the small world (Watts--Strogatz) models,
corresponding to the functions {\tt pref} and {\tt smallw}
of the CONTEST toolbox 
for Matlab. Good agreement between the two ranking methods
was observed on the networks obtain with the preferential
attachment method, especially as the density of the graphs
increased. More pronounced differences
between the rankings produced with the two methods were observed
in the case of small world networks.
Overall, the two importance 
rankings matched more closely when the matrix exponential 
was used than when under the matrix resolvent.  
 
We also presented the results of experiments with
real-world networks including social networks, citation 
networks, PPI networks, and infrastructure (transportation)
networks.
Here we found that overall, the two (complete) sets 
of rankings were closer to each other when the matrix resolvent was 
used instead of the matrix exponential.  However, when only 
the top 1\% of nodes was examined, the rankings matched more 
closely when the matrix exponential was used.  
This suggests 
that, for the networks tested, the resolvent-based rankings 
match more closely on ``unimportant'' (low-ranked) nodes and 
the exponential-based rankings exhibit more agreement on the 
``important'' (top-ranked) nodes.
 
In general, there is no simple way to compare two ranking schemes 
and determine that one is ``better'' than the other.  However, 
the total communicability rankings take into account 
more of the network topology than the subgraph centrality rankings 
(all walks starting at node $i$ versus all closed walks starting 
at node $i$).  This added information often (but not always) changes 
the ranking of the nodes to a certain degree, although there are 
many cases where there is still a strong similarity between the 
two sets of rankings.  The main benefit of using total 
communicability to rank the nodes is that the ranking can be 
estimated extremely quickly using Krylov subspace methods.
Indeed, as the Wikipedia graph calculation described in
section \ref{sec:large} shows,
for very large networks only the total communicability
(row sum) method is computationally feasible, the
subgraph centrality ranking being prohibitively expensive 
to compute.
Even if total communicability cannot always be recommended as a 
cheaper alternative to subgraph centrality, it provides valuable 
information about the network and can be used along
with other ranking schemes.

Finally, we have introduced the total communicability of a network
as a global measure of connectivity and of the ease of information
flow on a given network. This measure can be computed quickly
even for very large networks, and could be of interest in the
design of communication networks.
 
\section*{Acknowledgments}
We are indebted to Prof.~Ernesto Estrada (University of
Strathclyde) for providing the Intravenous Drug User and
Yeast PPI network data, and to Mr.~Yu Wang (Emory University)
for performing the calculations with the Wikipedia graph.\\


\end{document}